\documentclass[a4paper,UKenglish,cleveref, autoref, thm-restate]{lipics-v2021}

\nolinenumbers

\usepackage{amsmath,amssymb,amsthm}
\usepackage[ruled,vlined,linesnumbered]{algorithm2e}

\newtheorem*{theorem*}{Theorem}

\title{MAX CUT in Weighted Random Intersection Graphs and Discrepancy of Sparse Random Set Systems} 

\titlerunning{MAX CUT in Weighted Random Intersection Graphs} 

\author{Sotiris Nikoletseas}{Computer Engineering \& Informatics Department, University of Patras, Greece \and Computer Technology Institute, Greece}{nikole@ceid.upatras.gr}{https://orcid.org/0000-0003-3765-5636}{}

\author{Christoforos Raptopoulos\footnote{Corresponding author}}{Computer Engineering \& Informatics Department, University of Patras, Greece}{raptopox@ceid.upatras.gr}{https://orcid.org/0000-0002-9837-2632}{Supported by the Hellenic Foundation for Research and Innovation (H.F.R.I.) under the ``2nd Call for H.F.R.I. Research Projects to support Post-Doctoral Researchers'' (Project Number:704).}

\author{Paul Spirakis}{Department of Computer Science, University of Liverpool, UK \and Computer Engineering \& Informatics Department, University of Patras, Greece \and Computer Technology Institute, Greece}{p.spirakis@liverpool.ac.uk}{https://orcid.org/0000-0001-5396-3749}{Supported by NeST initiative of the School of EEE and CS at the U. of Liverpool and
by the EPSRC grant EP/P02002X/1 
}

\authorrunning{S. Nikoletseas, C. Raptopoulos, P. Spirakis} 

\Copyright{Sotiris Nikoletseas, Christoforos Raptopoulos and Paul Spirakis} 

\ccsdesc[500]{Mathematics of computing~Random graphs} 

\keywords{Random Intersection Graphs, Maximum Cut, Discrepancy} 

\category{} 

\relatedversion{} 

\begin{document}

\maketitle

\begin{abstract}
Let $V$ be a set of $n$ vertices, ${\cal M}$ a set of $m$ labels, and let $\mathbf{R}$ be an $m \times n$ matrix of independent Bernoulli random variables with probability of success $p$; columns of $\mathbf{R}$ are incidence vectors of label sets assigned to vertices. A random instance $G(V, E, \mathbf{R}^T \mathbf{R})$ of the weighted random intersection graph model is constructed by drawing an edge with weight equal to the number of common labels (namely $[\mathbf{R}^T \mathbf{R}]_{v,u}$) between any two vertices $u, v$ for which this weight is strictly larger than 0. In this paper we study the average case analysis of \textsc{Weighted Max Cut}, assuming the input is a weighted random intersection graph, i.e. given $G(V, E, \mathbf{R}^T \mathbf{R})$ we wish to find a partition of $V$ into two sets so that the total weight of the edges having exactly one endpoint in each set is maximized.

In particular, we initially prove that the weight of a maximum cut of $G(V, E, \mathbf{R}^T \mathbf{R})$ is concentrated around its expected value, and then show that, when the number of labels is much smaller than the number of vertices (in particular, $m=n^{\alpha}, \alpha<1$), a random partition of the vertices achieves asymptotically optimal cut weight with high probability. Furthermore, in the case $n=m$ and constant average degree (i.e. $p = \frac{\Theta(1)}{n}$), we show that with high probability, a majority type randomized algorithm outputs a cut with weight that is larger than the weight of a random cut by a multiplicative constant strictly larger than 1. Then, we formally prove a connection between the computational problem of finding a (weighted) maximum cut in $G(V, E, \mathbf{R}^T \mathbf{R})$ and the problem of finding a 2-coloring that achieves minimum discrepancy for a set system $\Sigma$ with incidence matrix $\mathbf{R}$ (i.e. minimum imbalance over all sets in $\Sigma$). We exploit this connection by proposing a (weak) bipartization algorithm for the case $m=n, p = \frac{\Theta(1)}{n}$ that, when it terminates, its output can be used to find a 2-coloring with minimum discrepancy in a set system with incidence matrix $\mathbf{R}$. In fact, with high probability, the latter 2-coloring corresponds to a bipartition with maximum cut-weight in $G(V, E, \mathbf{R}^T \mathbf{R})$. Finally, we prove that our (weak) bipartization algorithm terminates in polynomial time, with high probability, at least when $p = \frac{c}{n}, c<1$.  
\end{abstract}

\section{Introduction}

Given an undirected graph $G (V,E)$, the \textsc{Max Cut} problem asks for a partition of the vertices of $G$ into two sets, such that the number of edges with exactly one endpoint in each set of the partition is maximized. This problem can be naturally generalized for weighted (undirected) graphs. A weighted graph is denoted by $G (V, E, \mathbf{W})$, where $V$ is the set of vertices, $E$ is the set of edges and $\mathbf{W}$ is a weight matrix, which specifies a weight $\mathbf{W}_{i,j}=w_{i,j}$, for each pair of vertices $i, j$. In particular, we assume that $\mathbf{W}_{i,j}=0$, for each edge $\{i,j\} \notin E$. 

\begin{definition}[\textsc{Weighted Max Cut}]
Given a weighted graph $G (V, E, \mathbf{W})$, find a partition of $V$ into two (disjoint) subsets $A, B$, so as to maximize the cumulative weight of the edges of $G$ having one endpoint in $A$ and the other in $B$.
\end{definition}

\textsc{Weighted Max Cut} is fundamental in theoretical computer science and is relevant in various graph layout and embedding problems \cite{DPS02}. Furthermore, it also has many practical applications, including infrastructure cost and circuit layout optimization in network and VLSI design \cite{PT95}, minimizing the Hamiltonian of a spin glass model in statistical physics \cite{BGJR88}, and data clustering \cite{PZ06}. In the worst case \textsc{Max Cut} (and also \textsc{Weighted Max Cut}) is \textsc{APX}-hard, meaning that there is no polynomial-time approximation scheme that finds a solution that is arbitrarily close to the optimum, unless P = NP \cite{PY91}.  

The average case analysis of \textsc{Max Cut}, namely the case where the input graph is chosen at random from a probabilistic space of graphs, is also of considerable interest and is further motivated by the desire to justify and understand why various graph partitioning heuristics work well in practical applications. In most research works the input graphs are drawn from the Erd\H{o}s-R\'{e}nyi random graphs model ${\cal G}_{n, m}$, i.e. random instances are drawn equiprobably from the set of simple undirected graphs on $n$ vertices and $m$ edges, where $m$ is a linear function of $n$ (see also \cite{GL18,CMS06} for the average case analysis of \textsc{Max Cut} and its generalizations with respect to other random graph models). One of the earliest results in this area is that \textsc{Max Cut} undergoes a phase transition on ${\cal G}_{n, \gamma n}$ at $\gamma=\frac{1}{2}$ \cite{CGHS04}, in that the difference between the number of edges of the graph and the Max-Cut size is $O(1)$, for $\gamma <\frac{1}{2}$, while it is $\Omega(n)$, when $\gamma > \frac{1}{2}$. For large values of $\gamma$, it was proved in \cite{BGT13} that the maximum cut size of $G_{n, \gamma n}$ normalized by the number of vertices $n$ reaches an absolute limit in probability as $n \to \infty$, but it was not until recently that the latter limit was established and expressed analytically in \cite{DMS17}, using the interpolation method; in particular, it was shown to be asymptotically equal to $(\frac{\gamma}{2}+P_* \sqrt{\frac{\gamma}{2}})n$, where $P_* \approx 0.7632$. We note however that these results are existential, and thus do not lead to an efficient approximation scheme for finding a tight approximation of the maximum cut with large enough probability when the input graph is drawn from ${\cal G}_{n, \gamma n}$. An efficient approximation scheme in this case was designed in \cite{CGHS04}, and it was proved that, with high probability, this scheme constructs a cut with at least $\left(\frac{\gamma}{2} + 0.37613 \sqrt{\gamma}\right)n = (1+0.75226 \frac{1}{\sqrt{\gamma}}) \frac{\gamma}{2}n$ edges, noting that $\frac{\gamma}{2}n$ is the size of a random cut (in which each vertex is placed independently and equiprobably in one of the two sets of the partition). Whether there exists an efficient approximation scheme that can close the gap between the approximation guarantee of \cite{CGHS04} and the limit of \cite{DMS17} remains an open problem.    

In this paper, we study the average case analysis of \textsc{Weighted Max Cut} when input graphs are drawn from the generalization of another well-established model of random graphs, namely the \emph{weighted random intersection graphs model} (the unweighted version of the model was initially defined in \cite{KSS99}). In this model, edges are formed through the intersection of label sets assigned to each vertex and edge weights are equal to the number of common labels between edgepoints. 

\begin{definition}[Weighted random intersection graph]
Consider a universe ${\cal M} = \{1, 2, \ldots, m\}$ of labels and a set of $n$ vertices $V$. We define the $m \times n$ representation matrix $\mathbf{R}$ whose entries are independent Bernoulli random variables with probability of success $p$. For $\ell \in {\cal M}$ and $v \in V$, we say that vertex $v$ has chosen label $\ell$ iff $\mathbf{R}_{\ell, v}=1$. Furthermore, we draw an edge with weight $[\mathbf{R}^T \mathbf{R}]_{v,u}$ between any two vertices $u, v$ for which this weight is strictly larger than 0.
The weighted graph $G = (V, E, \mathbf{R}^T \mathbf{R})$ is then a random instance of the weighted random intersection graphs model $\overline{\cal G}_{n, m, p}$.  
\end{definition}  

Random intersection graphs are relevant to and capture quite nicely social networking; vertices are the individual actors and labels correspond to specific types of interdependency. Other applications include oblivious resource sharing in a (general) distributed setting, efficient and secure communication in sensor networks [20], interactions of mobile agents traversing the web etc. (see e.g. the survey papers \cite{survey1,NRS18} for further motivation and recent research related to random intersection graphs). In all these settings, weighted random intersection graphs, in particular, also capture the strength of connections between actors (e.g. in a social network, individuals having several characteristics in common have more intimate relationships than those that share only a few common characteristics). One of the most celebrated results in this area is equivalence (measured in terms of total variation distance) of random intersection graphs and Erd\H{o}s-R\'enyi random graphs when the number of labels satisfies $m = n^{\alpha}, \alpha>6$ \cite{FSS00}. This bound on the number of labels was improved in \cite{R11}, by showing equivalence of sharp threshold functions among the two models for $\alpha \geq 3$. Similarity of the two models has been proved even for smaller values of $\alpha$ (e.g. for any $\alpha > 1$) in the form of various translation results (see e.g. Theorem 1 in \cite{RS05}), suggesting that some algorithmic ideas developed for Erd\H{o}s-R\'{e}nyi random graphs also work for random intersection graphs (and also weighted random intersection graphs). 





In view of this, in the present paper we study the average case analysis of \textsc{Weighted Max Cut} under the weighted random intersection graphs model, for the range $m=n^{\alpha}, \alpha \leq 1$ for two main reasons: First, the average case analysis of \textsc{Max Cut} has not been considered in the literature so far when the input is a drawn from the random intersection graphs model, and thus the asymptotic behaviour of the maximum cut remains unknown especially for the range of values where random intersection graphs and Erd\H{o}s-R\'enyi random graphs differ the most. Furthermore, studying a model where we can implicitly control its intersection number (indeed $m$ is an obvious upper bound on the number of cliques that can cover all edges of the graph) may help understand algorithmic bottlenecks for finding maximum cuts in Erd\H{o}s-R\'{e}nyi random graphs.

Second, we note that the representation matrix $\mathbf{R}$ of a weighted random intersection graph can be used to define a random set system $\Sigma$ consisting of $m$ sets $\Sigma=\{L_1, \ldots, L_m\}$, where $L_{\ell}$ is the set of vertices that have chosen label $\ell$; we say that $\mathbf{R}$ is the \emph{incidence matrix} of $\Sigma$. Therefore, there is a natural connection between \textsc{Weighted Max Cut} and the \textsc{discrepancy} of such random set systems, which we formalize in this paper. In particular, given a set system $\Sigma$ with incidence matrix $\mathbf{R}$, its \emph{discrepancy} is defined as $\text{disc}(\Sigma) = \min_{\mathbf{x} \in \{\pm 1\}^n} \max_{L \in \Sigma} \left| \sum_{v \in L} x_v \right| = \| \mathbf{R} \mathbf{x} \|_{\infty}$, i.e. it is the minimum imbalance of all sets in $\Sigma$ over all 2-colorings $\mathbf{x}$. Recent work on the discrepancy of random rectangular matrices defined as above \cite{AN21} has shown that, when the number of labels (sets) $m$ satisfies $n \geq 0.73 m \log{m}$, the discrepancy of $\Sigma$ is at most 1 with high probability. The proof of the main result in \cite{AN21} is based on a conditional second moment method combined with Stein's method of exchangeable pairs, and improves upon a Fourier analytic result of \cite{HR19}, and also upon previous results in \cite{EL16}, \cite{P18}. The design of an efficient algorithm that can find a 2-coloring having discrepancy $O(1)$ in this range still remains an open problem. Approximation algorithms for a similar model for random set systems were designed and analyzed in \cite{BM19}; however, the algorithmic ideas there do not apply in our case.


\subsection{Our Contribution}

In this paper, we introduce the model of weighted random intersection graphs and we study the average case analysis of \textsc{Weighted Max Cut} through the prism of \textsc{Discrepancy} of random set systems. We formalize the connection between these two combinatorial problems for the case of arbitrary weighted intersection graphs in Corollary \ref{thm:MaxCutvsDiscrepancy}. We prove that, given a weighted intersection graph $G = (V,E,\mathbf{R}^T \mathbf{R})$ with representation matrix $\mathbf{R}$, and a set system with incidence matrix $\mathbf{R}$, such that $\text{disc}(\Sigma) \leq 1$, a 2-coloring has maximum cut weight in $G$ if and only if it achieves minimum discrepancy in $\Sigma$. In particular, Corollary \ref{thm:MaxCutvsDiscrepancy} applies in the range of values considered in \cite{AN21} (i.e. $n \geq 0.73 m \log{m}$), and thus any algorithm that finds a maximum cut in $G(V,E,\mathbf{R}^T \mathbf{R})$ with large enough probability can also be used to find a 2-coloring with minimum discrepancy in a set system $\Sigma$ with incidence matrix $\mathbf{R}$, with the same probability of success.

We then consider weighted random intersection graphs in the case $m = n^{\alpha}, \alpha \leq 1$, and we prove that the maximum cut weight of a random instance $G(V,E,\mathbf{R}^T \mathbf{R})$ of $\overline{{\cal G}}_{n, m, p}$ concentrates around its expected value (see Theorem \ref{theorem:MaxCutconcentration}). In particular, with high probability (whp, i.e. with probability tending to 1 as $n \to \infty$) over the choices of $\mathbf{R}$, $\texttt{Max-Cut}(G) \sim \mathbb{E}_{\mathbf{R}}[\texttt{Max-Cut}(G)]$, where $\mathbb{E}_{\mathbf{R}}$ denotes expectation with respect to $\mathbf{R}$. The proof is based on the Efron-Stein inequality for upper bounding the variance of the maximum cut. As a consequence of our concentration result, we prove in Theorem \ref{theorem:randomcut} that, in the case $\alpha <1$, a random 2-coloring (i.e. biparition) $\mathbf{x}^{(rand)}$ in which each vertex chooses its color independently and equiprobably, has cut weight asymptotically equal to $\texttt{Max-Cut}(G)$, with high probability over the choices of $\mathbf{x}^{(rand)}$ and $\mathbf{R}$. 

The latter result on random cuts allows us to focus the analysis of our randomized algorithms of Section \ref{sec:randomizedalgorithms} on the case $m=n$ (i.e. $\alpha=1$), and $p = \frac{c}{n}$, for some constant $c$ (see also the discussion at the end of subsection \ref{sec:randomcut}), where the assumptions of Theorem \ref{theorem:randomcut} do not hold. It is worth noting that, in this range of values, the expected weight of a fixed edge in a weighted random intersection graph is equal to $mp^2 = \Theta(1/n)$, and thus we hope that our work here will serve as an intermediate step towards understanding when algorithmic bottlenecks for \textsc{Max Cut} appear in sparse random graphs (especially Erd\H{o}s-R\'{e}nyi random graphs) with respect to the intersection number. In particular, we analyze a Majority Cut Algorithm \ref{alg:majoritycut} that extends the algorithmic idea of \cite{CGHS04} to weighted intersection graphs as follows: vertices are colored sequentially (each color $+1$ or $-1$ corresponding to a different set in the partition of the vertices), and the $t$-th vertex is colored opposite to the sign of $\sum_{i \in [t-1]} [\mathbf{R}^T \mathbf{R}]_{i,t} x_i$, namely the total available weight of its incident edges, taking into account colors of adjacent vertices. Our average case analysis of the Majority Cut Algorithm shows that, when $m=n$ and $p = \frac{c}{n}$, for large constant $c$, with high probability over the choices of $\mathbf{R}$, the expected weight of the constructed cut is at least $1+\beta$ times larger than the expected weight of a random cut, for some constant $\beta = \beta(c) \geq \sqrt{\frac{16}{27 \pi c^3}} - o(1)$. The fact that the lower bound on beta is inversely proportional to $c^{3/2}$ was to be expected, because, as $p$ increases, the approximation of the maximum cut that we get from the weight of a random cut improves (see also the discussion at the end of subsection \ref{sec:randomcut}). 

In subsection \ref{sec:weakBipartization} we propose a framework for finding maximum cuts in weighted random intersection graphs for $m=n$ and $p = \frac{c}{n}$, for  constant $c$, by exploiting the connection between \textsc{Weighted Max Cut} and the problem of discrepancy minimization in random set systems. In particular, we design a Weak Bipartization Algorithm \ref{alg:WeakBipartization}, that takes as input an intersection graph with representation matrix $\mathbf{R}$ and outputs a subgraph that is ``almost'' bipartite. In fact, the input intersection graph is treated as a multigraph composed by overlapping cliques formed by the label sets $L_{\ell} = \{v: \mathbf{R}_{\ell, v}=1\}, \ell \in {\cal M}$. The algorithm attempts to destroy all odd cycles of the input (except from odd cycles that are formed by labels with only two vertices) by replacing each clique induced by some label set $L_{\ell}$ by a random maximal matching. In Theorem \ref{theorem:discrepancytomaxcut} we prove that, with high probability over the choices of $\mathbf{R}$, if the Weak Bipartization Algorithm terminates, then its output can be used to construct a 2-coloring that has minimum discrepancy in a set system with incidence matrix $\mathbf{R}$, which also gives a maximum cut in $G(V,E,\mathbf{R}^T \mathbf{R})$. It is worth noting that this does not follow from Corollary \ref{thm:MaxCutvsDiscrepancy}, because a random set system with incidence matrix $\mathbf{R}$ has discrepancy larger than 1 with (at least) constant probability when $m=n$ and $p = \frac{c}{n}$. Our proof relies on a structural property of closed 0-strong vertex-label sequences (loosely defined as closed walks of edges formed by distinct labels) in the weighted random intersection graph $G(V, E, \mathbf{R}^T \mathbf{R})$ (Lemma \ref{lem:disjointcycles}). Finally, in Theorem \ref{theorem:c<1termination}, we prove that our Weak Bipartization Algorithm terminates in polynomial time, with high probability, if the constant $c$ is strictly less than 1. Therefore, there is a polynomial time algorithm for finding weighted maximum cuts, with high probability, when the input is drawn from $\overline{{\cal G}}_{n, n, \frac{c}{n}}$, with $c<1$. We believe that this part of our work may also be of interest regarding the design of efficient algorithms for finding minimum disrepancy colorings in random set systems. 

Due to lack of space, some of the proofs are given in a clearly marked Appendix, to be read at the discretion of the program committee.

\section{Notation and preliminary results}

We denote weighted undirected graphs by $G(V, E, \mathbf{W})$; in particular, $V=V(G)$ (resp. $E=E(G)$) is the set of vertices (resp. set of edges) and $\mathbf{W} = \mathbf{W}(G)$ is the weight matrix, i.e. $\mathbf{W}_{i, j} = w_{i, j}$ is the weight of (undirected) edge $\{i,j\} \in E$. We allow $\mathbf{W}$ to have non-zero diagonal entries, as these do not affect cut weights. We also denote the number of vertices by $n$, and we use the notation $[n] = \{1,2,\ldots,n\}$. We also use this notation to define parts of matrices, for example $\mathbf{W}_{[n], 1}$ denotes the first column of the weight matrix.

A bipartition of the sets of vertices is a partition of $V$ into two sets $A, B$ such that $A \cap B = \emptyset$ and $A \cup B = V$. Bipartitions correspond to 2-colorings, which we denote by vectors $\mathbf{x}$ such that $x_i=+1$ if $i \in A$ and $x_i=-1$ if $i \in B$.

Given a weighted graph $G(V, E, \mathbf{W})$, we denote by $\texttt{Cut}(G, \mathbf{x})$ the weight of a cut defined by a bipartition $\mathbf{x}$, namely $\texttt{Cut}(G, \mathbf{x}) = \sum_{\{i, j\} \in E: i \in A, j \in B} w_{i,j} = \frac{1}{4} \sum_{\{i, j\} \in E} w_{i,j} (x_i-x_j)^2$. The maximum cut of $G$ is $\texttt{Max-Cut}(G) = \max_{\mathbf{x} \in \{-1, +1\}^n} \texttt{Cut}(G, \mathbf{x})$.

For a weighted random intersection graph $G(V,E, \mathbf{R}^T \mathbf{R})$ with representation matrix $\mathbf{R}$, we denote by $S_v$ the set of labels chosen by vertex $v \in V$, i.e. $S_v = \{\ell: \mathbf{R}_{\ell, v}=1\}$. Furthermore, we denote by $L_{\ell}$ the set of vertices having chosen label $\ell$, i.e. $L_{\ell}=\{v:\mathbf{R}_{\ell, v}=1\}$. Using this notation, the weight of an edge $\{v, u\} \in E$ is $|S_v \cap S_u|$; notice also that this is equal to 0 when $\{v, u\} \notin E$. We also note here that we may also think of a weighted random intersection graph as a simple weighted graph where, for any pair of vertices $v, u$, there are $|S_v \cap S_u|$ simple edges between them.

A set system $\Sigma$ defined on a set $V$ is a family of sets $\Sigma = \{L_1, L_2, \ldots, L_m\}$, where $L_\ell \subseteq V, \ell \in [m]$. The incidence matrix of $\Sigma$ is an $m \times n$ matrix $\mathbf{R} = \mathbf{R}(\Sigma)$, where for any $\ell \in [m], v \in [n]$, $\mathbf{R}_{\ell, v} = 1$ if $v \in S_{\ell}$ and 0 otherwise. The discrepancy of $\Sigma$ with respect to a 2-coloring $\mathbf{x}$ of the vertices in $V$ is $\text{disc}(\Sigma, \mathbf{x}) = \max_{\ell \in [m]} \left| \sum_{v \in V} \mathbf{R}_{\ell, v} x_v \right| = \| \mathbf{R} \mathbf{x} \|_{\infty}$. The discrepancy of $\Sigma$ is $\text{disc}(\Sigma) = \min_{\mathbf{x} \in \{-1, +1\}^n} \text{disc}(\Sigma, \mathbf{x})$.


It is well-known that the cut size of a bipartition of the set of vertices of a graph $G(V,E)$ into sets $A$ and $B$ is given by $\frac{1}{4} \sum_{\{i,j\} \in E} (x_i-x_j)^2$, where $x_i=+1$ if $i \in A$ and $x_i=-1$ if $i \in B$. This can be naturally generalized for multigraphs and also for weighted graphs. In particular, the \texttt{Max-Cut} size of a weighted graph $G(V, E, \mathbf{W})$ is given by
\begin{equation} \label{eq:maxcut}
\texttt{Max-Cut}(G) = \max_{\mathbf{x} \in \{-1, +1\}^n} \frac{1}{4} \sum_{\{i,j\} \in E} \mathbf{W}_{i,j} (x_i-x_j)^2.
\end{equation} 

In particular, we get the following Corollary (refer to Section \ref{sec:corollary1proof} of the Appendix for the proof):

\begin{corollary} \label{corollary:maxcutRIGformula}
Let $G(V,E, \mathbf{R}^T \mathbf{R})$ be a weighted intersection graph with representation matrix $\mathbf{R}$. Then, for any $\mathbf{x} \in \{-1, +1\}^n$,
\begin{equation} \label{eq:cutRIGformula}
\texttt{Cut}(G, \mathbf{x}) = \frac{1}{4} \left(\sum_{i,j \in [n]^2} \left[\mathbf{R}^T \mathbf{R} \right]_{i,j} - \left\| \mathbf{R} \mathbf{x} \right\|^2 \right)
\end{equation}
and so
\begin{equation} \label{eq:maxcutRIGformula}
\texttt{Max-Cut}(G) = \frac{1}{4} \left(\sum_{i,j \in [n]^2} \left[\mathbf{R}^T \mathbf{R} \right]_{i,j} - \min_{\mathbf{x} \in \{-1, +1\}^n}  \left\| \mathbf{R} \mathbf{x} \right\|^2 \right),
\end{equation}
where $\|\cdot\|$ denotes the 2-norm. In particular, the expectation of the size of a random cut, where each entry of $\mathbf{x}$ is independently and equiprobably either +1 or -1 is equal to $\mathbb{E}_{\mathbf{x}}\left[ \texttt{Cut}(G, \mathbf{x})\right] = \frac{1}{4} \sum_{i\neq j, i,j \in [n]} \left[\mathbf{R}^T \mathbf{R} \right]_{i,j}$, where $\mathbb{E}_{\mathbf{x}}$ denotes expectation with respect to $\mathbf{x}$. 
\end{corollary}

Since $\sum_{i,j \in [n]^2} \left[\mathbf{R}^T \mathbf{R} \right]_{i,j}$ is fixed for any given representation matrix $\mathbf{R}$, the above Corollary implies that, to find a bipartition of the vertex set $V$ that corresponds to a maximum cut, we need to find an $n$-dimensional vector in $\arg \min_{\mathbf{x} \in \{-1, +1\}^n}  \left\| \mathbf{R} \mathbf{x} \right\|^2$. We thus get the following (refer to Section \ref{sec:previoustheorem} of the Appendix for the proof):

\begin{corollary} \label{thm:MaxCutvsDiscrepancy}
Let $G(V,E, \mathbf{R}^T \mathbf{R})$ be a weighted intersection graph with representation matrix $\mathbf{R}$ and $\Sigma$ a set system with incidence matrix $\mathbf{R}$. If $\text{disc}(\Sigma) \leq 1$, then $\mathbf{x}^* \in \arg \min_{\mathbf{x} \in \{-1, +1\}^n}  \left\| \mathbf{R} \mathbf{x} \right\|^2$ if and only if $\mathbf{x}^* \in \arg \min_{\mathbf{x} \in \{-1, +1\}^n}  \text{disc}(\Sigma, \mathbf{x})$. In particular, if the minimum discrepancy of $\Sigma$ is at most 1, a bipartition corresponds to a maximum cut iff it achieves minimum discrepancy.
\end{corollary}  

Notice that above result is not necessarily true when $\text{disc}(\Sigma) > 1$, since the minimum of $\| \mathbf{R} \mathbf{x} \|$ could be achieved by 2-colorings with larger discrepancy than the optimal.

\subsection{Range of values for $p$}

Concerning the success probability $p$, we note that, when $p = o\left( \sqrt{\frac{1}{nm}} \right)$, direct application of the results of \cite{BTU09} suggest that $G(V,E, \mathbf{R}^T \mathbf{R})$ is chordal with high probability, but in fact the same proofs reveal that a stronger property holds, namely that there is no closed vertex-label sequence (refer to the precise definition in subsection \ref{sec:weakBipartization}) having distinct labels. Therefore, in this case, finding a bipartition with maximum cut weight is straightforward: indeed, one way to construct a maximum cut is to run our Weak Bipartization Algorithm \ref{alg:WeakBipartization} from subsection \ref{sec:weakBipartization}, and then to apply Theorem \ref{theorem:discrepancytomaxcut} (noting that Weak Bipartization termination condition trivially holds, since the set ${\cal C}_{odd}(G^{(b)})$ defined in subsection \ref{sec:weakBipartization} is empty). Furthermore, even though we consider weighted graphs, we will also assume that $mp^2 = O(1)$, noting that, otherwise, $G(V,E, \mathbf{R}^T \mathbf{R})$ will be almost complete with high probability (indeed, the unconditional edge existence probability is $1-(1-p^2)^m$, which tends to 1 for $mp^2 = \omega(1)$). In particular, we will assume that $C_1 \sqrt{\frac{1}{nm}} \leq p \leq C_2 \frac{1}{\sqrt{m}}$, for arbitrary positive constants $C_1, C_2$; $C_1$ can be as small as possible, and $C_2$ can be as large as possible, provided $C_2 \frac{1}{\sqrt{m}} \leq 1$. We note that, when $p$ is asymptotically equal to the upper bound $C_2 \frac{1}{\sqrt{m}}$, there is no constant weight upper bound that holds with high probability, whereas, when $p$ is asymptotically equal to the lower bound $C_1 \sqrt{\frac{1}{nm}}$, all weights in the graph are bounded by a small constant with high probability. Our results in Section \ref{sec:concentration} assume this range of values for $p$, and thus graph instances may contain edges with large (but constant) weights. On the other hand, in the analysis of our randomized algorithms in section \ref{sec:randomizedalgorithms}, we assume $n=m$ and $p = \Theta\left(\frac{1}{n} \right)$; this range of values gives sparse graph instances (even though the distribution is different from sparse Erd\H{o}s-R\'{e}nyi random graphs).


\section{Concentration of Max-Cut} \label{sec:concentration}

In this section we prove that the size of the maximum cut in a weighted random intersection graph concentrates around its expected value. We note however, that the following Theorem does not provide an explicit formula for the expected value of the maximum cut. 

\begin{theorem} \label{theorem:MaxCutconcentration}
Let $G(V,E, \mathbf{R}^T \mathbf{R})$ be a random instance of the $\overline{{\cal G}}_{n, m, p}$ model with $m=n^a, \alpha \leq 1$, and $C_1 \sqrt{\frac{1}{nm}} \leq p \leq 1$, for arbitrary positive constant $C_1$, and let $\mathbf{R}$ be its representation matrix. Then $\texttt{Max-Cut}(G) \sim \mathbb{E}_{\mathbf{R}}[\texttt{Max-Cut}(G)]$ with high probability, where $\mathbb{E}_{\mathbf{R}}$ denotes expectation with respect to $\mathbf{R}$, i.e. $\texttt{Max-Cut}(G)$ concentrates around its expected value. 
\end{theorem}
\begin{proof} Let $G=G(V,E, \mathbf{R}^T \mathbf{R})$ be a weighted random intersection graph, and let $\mathbf{D}$ denote the (random) diagonal matrix containing all diagonal elements of $\mathbf{R}^T\mathbf{R}$. In particular, equation (\ref{eq:maxcutRIGformula}) of Corollary \ref{corollary:maxcutRIGformula} can be written as 
\begin{equation*}
\texttt{Max-Cut}(G) = \frac{1}{4} \left(\sum_{i\neq j, i,j \in [n]} \left[\mathbf{R}^T \mathbf{R} \right]_{i,j} - \min_{\mathbf{x} \in \{-1, +1\}^n}  \mathbf{x}^T \left(\mathbf{R}^T \mathbf{R} -\mathbf{D}\right) \mathbf{x} \right).
\end{equation*}
Furthermore, for any given $\mathbf{R}$, notice that, if we select each element of $\mathbf{x}$ independently and equiprobably from $\{-1, +1\}$, then $\mathbb{E}_{\mathbf{x}}[\mathbf{x}^T \left(\mathbf{R}^T \mathbf{R} -\mathbf{D}\right) \mathbf{x}]=0$, where $\mathbb{E}_{\mathbf{x}}$ denotes expectation with respect to $\mathbf{x}$. Therefore, by the probabilistic method, $\min_{\mathbf{x} \in \{-1, +1\}^n}  \mathbf{x}^T \left(\mathbf{R}^T \mathbf{R} -\mathbf{D}\right) \mathbf{x} \leq 0$, implying the following bound:
\begin{equation} \label{eq:maxcutbounds}
\frac{1}{4} \sum_{i\neq j, i,j \in [n]} \left[\mathbf{R}^T \mathbf{R} \right]_{i,j} \leq \texttt{Max-Cut}(G) \leq \frac{1}{2} \sum_{i\neq j, i,j \in [n]} \left[\mathbf{R}^T \mathbf{R} \right]_{i,j},
\end{equation}
where the second inequality follows trivially by observing that $\frac{1}{2} \sum_{i\neq j, i,j \in [n]} \left[\mathbf{R}^T \mathbf{R} \right]_{i,j}$ equals the sum of the weights of all edges.

By linearity, $\mathbb{E}_{\mathbf{R}}\left[\sum_{i\neq j, i,j \in [n]} \left[\mathbf{R}^T \mathbf{R} \right]_{i,j} \right] = \mathbb{E}_{\mathbf{R}}\left[\sum_{i\neq j, i,j \in [n]} \sum_{\ell \in [m]} \mathbf{R}_{\ell,i} \mathbf{R}_{\ell, j} \right] = n(n-1)mp^2 = \Theta(n^2mp^2)$, which goes to infinity as $n \to \infty$, because $np = \Omega\left(\sqrt{\frac{n}{m}} \right) = \Omega(1)$ in the range of parameters that we consider. In particular, by (\ref{eq:maxcutbounds}), we have 
\begin{equation} \label{eq:expectationasymptotics}
\mathbb{E}_{\mathbf{R}}[\texttt{Max-Cut}(G)] = \Theta(n^2mp^2).
\end{equation}
By Chebyshev's inequality, for any $\epsilon>0$, we have
\begin{equation} \label{eq:Chebyshevbound}
\Pr\left(\left| \texttt{Max-Cut}(G)-\mathbb{E}_{\mathbf{R}}[\texttt{Max-Cut}(G)]\right| \geq \epsilon n^2mp^2 \right) \leq \frac{\text{Var}_{\mathbf{R}}(\texttt{Max-Cut}(G))}{\epsilon^2 n^4m^2p^4},
\end{equation}
where $\text{Var}_{\mathbf{R}}$ denotes variance with respect to $\mathbf{R}$. To bound the variance on the right hand side of the above inequality, we use the Efron-Stein inequality. In particular, we write $\texttt{Max-Cut}(G):= f(\mathbf{R})$, i.e. we view $\texttt{Max-Cut}(G)$ as a function of the label choices. For $\ell \in [m], i \in [n]$, we also write $\mathbf{R}^{(\ell, i)}$ for the matrix $\mathbf{R}$ where entry $(\ell, i)$ has been replaced by an independent, identically distributed (i.i.d.) copy of $\mathbf{R}_{\ell, i}$, which we denote by $\mathbf{R}_{\ell, i}'$. By the Efron-Stein inequality, we have
\begin{equation} \label{eq:EfronStein}
\text{Var}_{\mathbf{R}}(\texttt{Max-Cut}(G)) \leq \frac{1}{2} \sum_{\ell \in [m], i \in [n]} \mathbb{E}\left[ \left( f(\mathbf{R}) - f\left(\mathbf{R}^{(\ell, i)} \right)\right)^2 \right].
\end{equation} 
Notice now that, given all entries of $\mathbf{R}$ except $\mathbf{R}_{\ell, i}$, the probability that $f(\mathbf{R})$ is different from $f\left(\mathbf{R}^{(\ell, i)} \right)$ is at most $\Pr(\mathbf{R}_{\ell, i} \neq \mathbf{R}_{\ell, i}') = 2p(1-p)$. Furthermore, if $L_{\ell} \backslash \{i\}$ is the set of vertices different from $i$ which have selected $\ell$, we then have that $\left( f(\mathbf{R}) - f\left(\mathbf{R}^{(\ell, i)} \right)\right)^2 \leq |L_{\ell} \backslash \{i\}|^2$, because the intersection graph with representation matrix $\mathbf{R}$ differs by at most $|L_{\ell} \backslash \{i\}|$ edges from the intersection graph with representation matrix $\mathbf{R}^{(\ell, i)}$. Also note that, by definition, $|L_{\ell} \backslash \{i\}|$ follows the Binomial distribution ${\cal B}(n-1, p)$. In particular, $\mathbb{E} \left[ |L_{\ell} \backslash \{i\}|^2 \right] = (n-1)p(np-2p+1)$, implying $\mathbb{E}\left[ \left( f(\mathbf{R}) - f\left(\mathbf{R}^{(\ell, i)} \right)\right)^2 \right] \leq 2p(1-p) (n-1)p(np-2p+1)$, for any fixed $\ell \in [m], i \in [n]$. 

Putting this all together, (\ref{eq:EfronStein}) becomes
\begin{eqnarray} 
\text{Var}_{\mathbf{R}}(\texttt{Max-Cut}(G)) & \leq & \frac{1}{2} \sum_{\ell \in [m], i \in [n]} 2p(1-p) (n-1)p(np-2p+1) \nonumber \\
& = & nm p(1-p) (n-1)p(np-2p+1) = O(n^3mp^3),
\end{eqnarray} 
Therefore, by (\ref{eq:Chebyshevbound}), we get 
\begin{displaymath} 
\Pr\left(\left| \texttt{Max-Cut}(G)-\mathbb{E}_{\mathbf{R}}[\texttt{Max-Cut}(G)]\right| \geq \epsilon n^2mp^2 \right) \leq \frac{O(n^3mp^3)}{\epsilon^2 n^4m^2p^4} = O\left( \frac{1}{\epsilon^2 nmp}\right),
\end{displaymath}
which goes to 0 in the range of values that we consider. Together with (\ref{eq:expectationasymptotics}), the above bound proves that $\texttt{Max-Cut}(G)$ is concentrated around its expected value. 
\end{proof}

\subsection{\texttt{Max-Cut} for small number of labels} \label{sec:randomcut}

Using Theorem \ref{theorem:MaxCutconcentration}, we can now show that, in the case $m = n^{\alpha}, \alpha<1$, and $p = O\left( \frac{1}{\sqrt{m}}\right)$, a random cut has asymptotically the same weight as $\texttt{Max-Cut}(G)$, where $G=G(V,E, \mathbf{R}^T \mathbf{R})$ is a random instance of $\overline{\cal G}_{n, m, p}$. In particular, let $\mathbf{x}^{(rand)}$ be constructed as follows: for each $i \in [n]$, set $x^{(rand)}_{i} = -1$ independently with probability $\frac{1}{2}$, and $x^{(rand)}_{i} = +1$ otherwise. 

The proof details of the following Theorem can be found in Section \ref{sec:theorem3proof} of the Appendix. In view of equation (\ref{eq:maxcutRIGformula}), the main idea is to prove that, with high probability over random $\mathbf{x}$ and $\mathbf{R}$, $\|\mathbf{R} \mathbf{x}\|^2$ is asymptotically smaller than the expectation of the weight of the cut defined by $\mathbf{x}^{(rand)}$, in which case the theorem follows by concentration of $\texttt{Max-Cut}(G)$ around its expected value (Theorem \ref{theorem:MaxCutconcentration}), and straightforward bounds on $\texttt{Max-Cut}(G)$.







\begin{theorem} \label{theorem:randomcut}
Let $G(V,E, \mathbf{R}^T \mathbf{R})$ be a random instance of the $\overline{{\cal G}}_{n, m, p}$ model with $m=n^a, \alpha < 1$, and $C_1 \sqrt{\frac{1}{nm}} \leq p \leq C_2 \frac{1}{\sqrt{m}}$, for arbitrary positive constants $C_1, C_2$, and let $\mathbf{R}$ be its representation matrix. Then the cut weight of the random 2-coloring $\mathbf{x}^{(rand)}$ satisfies $\texttt{Cut}(G, \mathbf{x}^{(rand)}) = (1-o(1)) \texttt{Max-Cut}(G)$ with high probability over the choices of $\mathbf{x}^{(rand)}$, $\mathbf{R}$.
\end{theorem}

We note that the same analysis also holds when $n=m$ and $p$ is sufficiently large (e.g. $p = \omega(\frac{\ln{n}}{n})$); more details can be found at the end of Section \ref{sec:theorem3proof} of the Appendix. In view of this, in the following sections we will only assume $m=n$ (i.e. $\alpha=1$) and also $p =  \frac{c}{n}$, for some positive constant $c$. Besides avoiding complicated formulae for $p$, the reason behind this assumption is that, in this range of values, the expected weight of a fixed edge in $G(V,E, \mathbf{R}^T \mathbf{R})$ is equal to $mp^2 = \Theta(1/n)$, and thus we hope that our work will serve as an intermediate step towards understanding algorithmic bottlenecks for finding maximum cuts in Erd\H{o}s-R\'{e}nyi random graphs $G_{n, c/n}$ with respect to their intersection number.

\section{Algorithmic results (randomized algorithms)} \label{sec:randomizedalgorithms}

\subsection{The Majority Cut Algorithm} \label{sec:MajorityAlgorithm}

In the following algorithm, the 2-coloring representing the bipartition of a cut is constructed as follows: initially, a small constant fraction $\epsilon$ of vertices are randomly placed in the two partitions, and then in each subsequent step, one of the remaining vertices is placed in the partition that maximizes the weight of incident edges with endpoints in the opposite partition. 


\begin{algorithm}[H] \label{alg:majoritycut}
\KwIn{$G(V,E, \mathbf{R}^T \mathbf{R})$ and its representation matrix $\mathbf{R} \in \{0,1\}^{m \times n}$}
\KwOut{Large cut 2-coloring $\mathbf{x} \in \{-1,+1\}^n$}

Let $v_1, \ldots, v_n$ an arbitrary ordering of vertices\;
\For{$t=1$ \KwTo $\epsilon n$}{
	Set $x_t$ to either $-1$ or $+1$ independently with equal probability\;
}
\For{$t=\epsilon n+1$ \KwTo $n$}{
	\eIf{$\sum_{i \in [t-1]} [\mathbf{R}^T \mathbf{R}]_{i,t} x_i \geq 0$}{
	$x_t=-1$\;}{
	$x_t=+1$\;}
}
\Return{$\mathbf{x}$\;}

\caption{Majority Cut}
\end{algorithm}


Clearly the Majority Algorithm runs in polynomial time in $n, m$. Furthermore, the following Theorem provides a lower bound on the expected weight of the cut constructed by the algorithm in the case $m=n$, $p = \frac{c}{n}$, for large constant $c$, and $\epsilon \to 0$. The full proof details can be found in Section \ref{sec:majoritycutanalysis} of the Appendix. 

\begin{theorem} \label{thm:majoritycutanalysis}
Let $G(V,E, \mathbf{R}^T \mathbf{R})$ be a random instance of the $\overline{{\cal G}}_{n, m, p}$ model, with $m=n$, and $p = \frac{c}{n}$, for large positive constant $c$, and let $\mathbf{R}$ be its representation matrix. Then, with high probability over the choices of $\mathbf{R}$, the majority algorithm constructs a cut with expected weight at least $(1+\beta) \frac{1}{4} \mathbb{E}\left[\sum_{i\neq j, i,j \in [n]} \left[\mathbf{R}^T \mathbf{R} \right]_{i,j} \right]$, where $\beta = \beta(c) \geq \sqrt{\frac{16}{27 \pi c^3}} - o(1)$ is a constant, i.e. at least $1+\beta$ times larger than the expected weight of a random cut. 
\end{theorem}
\begin{proof}[Proof sketch] Let $G(V,E, \mathbf{R}^T \mathbf{R})$ be a random instance of the $\overline{{\cal G}}_{n, m, p}$ model, with $m=n$, and $p = \frac{c}{n}$, for some large enough constant $c$. For $t \in [n]$, let $M_t$ denote the constructed cut size just after the consideration of a vertex $v_t$, for some $t \geq \epsilon n+1$. By equation (\ref{eq:maxcutRIGformula}) for $n=t$, and since the values $x_1, \ldots, x_{t-1}$ are already decided in previous steps, we have
$M_t = \frac{1}{4} \left(\sum_{i,j \in [t]^2} \left[\mathbf{R}^T \mathbf{R} \right]_{i,j} - \min_{x_t \in \{-1, +1\}}  \left\| \mathbf{R}_{[m], [t]} \mathbf{x}_{[t]} \right\|^2 \right)$, and after careful calculation we get the recurrence
\begin{eqnarray*}
M_t & = & M_{t-1} + \frac{1}{2} \sum_{i \in [t-1]} \left[\mathbf{R}^T \mathbf{R} \right]_{i,t} + \frac{1}{2} \left| Z_t\right|,  
\end{eqnarray*}
where $Z_t = Z_t(\mathbf{x}, \mathbf{R}) = \sum_{i \in [t-1]} \left[\mathbf{R}^T \mathbf{R} \right]_{i,t} x_i = \sum_{\ell \in [m]} \mathbf{R}_{\ell, t} \sum_{i \in [t-1]}  \mathbf{R}_{\ell,i} x_i$. Observe that, in the latter recursive equation, the term $\frac{1}{2} \sum_{i \in [t-1]} \left[\mathbf{R}^T \mathbf{R} \right]_{i,t}$ corresponds to the expected increment of the constructed cut if the $t$-vertex chose its color uniformly at random. Therefore, lower bounding the expectation of $\frac{1}{2}\left| Z_t\right|$ will tell us how much better the Majority Algorithm does when considering the $t$-th vertex.

Towards this end, we note that, given $\mathbf{x}_{[t-1]} = \{x_i, i \in [t-1] \}$, and $\mathbf{R}_{[m], [t-1]}=\{ \mathbf{R}_{\ell, i}, \ell \in [m], i \in [t-1]\}$, $Z_t$ is the sum of $m$ independent random variables, since the Bernoulli random variables $\mathbf{R}_{\ell,t}, \ell \in [m],$ are independent, for any given $t$ (note that the conditioning is essential for independence, otherwise the inner sums in the definition of $Z_t$ would also depend on the $x_i$'s, which are not random when $i$ is large). By using a domination argument, we can then prove that 
\begin{equation*}
\mathbb{E}[|Z_t| \big| \mathbf{x}_{[t-1]}, \mathbf{R}_{[m], [t-1]}] \geq \mathrm{MD}(Z^B_t), 
\end{equation*}
where $Z^B_t$ is a certain Binomial random variable (formally defined in the full proof), and $\mathrm{MD}(\cdot)$ is the mean absolute difference of (two independent copies of) $Z^B_t$, namely $\mathrm{MD}(Z^B_t) = \mathbb{E}[\left|Z^B_t - Z'^B_t \right|]$. Even though we are aware of no simple closed formula for $\mathrm{MD}(Z^B_t)$, we resort to Gaussian approximation of $Z^B_t - Z'^B_t$ through the Berry-Esseen Theorem, ultimately showing that $|Z^B_t - Z'^B_t|$ follows approximately the \emph{folded normal distribution}. In particular, we show that $\mathrm{MD}(Z^B_t) \geq \sqrt{\frac{c (t-1)}{3 \pi n}} -o(1)$, and since the right hand side is independent of $\mathbf{x}_{[t-1]}, \mathbf{R}_{[m], [t-1]}$, we get the same lower bound on the expectation of $|Z_t|$, namely, $\mathbb{E}[|Z_t|] \geq \sqrt{\frac{c (t-1)}{3 \pi n}} -o(1)$. Summing over all $t \geq \epsilon n+1$, we get 
\begin{equation*}
\sum_{t \geq \epsilon n+1} \mathbb{E}\left[|Z_t| \right] \geq \sqrt{\frac{c}{3 \pi}} \left(\frac{2}{3} - \epsilon ^{3/2}\right) n -o(n),
\end{equation*} 
and the result follows by noting that the expected weight of a random cut is equal to $\frac{1}{4} n(n-1)mp^2 = \frac{c^2}{4}n + o(n)$, and taking $\epsilon \to 0$.


\end{proof}

\subsection{Intersection graph (weak) bipartization} \label{sec:weakBipartization}

Notice that we can view a weighted intersection graph $G(V, E, \mathbf{R}^T\mathbf{R})$ as a multigraph, composed by $m$ (possibly) overlapping cliques corresponding to the sets of vertices having chosen a certain label, namely $L_{\ell} = \{v: \mathbf{R}_{\ell, v}\}, \ell \in [m]$. In particular, let $K^{(\ell)}$ denote the clique induced by label $\ell$. Then $G = \cup^+_{\ell \in [m]} K^{(\ell)}$, where $\cup^+$ denotes union that keeps multiple edges. In this section, we present an algorithm that takes as input an intersection graph $G$ given as a union of overlapping cliques and outputs a subgraph that is ``almost'' bipartite. 


To facilitate the presentation of our algorithm, we first give some useful definitions. A \emph{closed vertex-label sequence} is a sequence of alternating vertices and labels starting and ending at the same vertex, namely $\sigma:= v_1, \ell_1, v_2, \ell_2, \cdots, v_k, \ell_{k}, v_{k+1}=v_1$, where the size of the sequence $k = |\sigma|$ is the number of its labels, $v_i \in V$, $\ell_i \in {\cal M}$, and $\{v_i, v_{i+1}\} \subseteq L_{\ell_i}$, for all $i \in [k]$ (i.e. $v_i$ is connected to $v_{i+1}$ in the intersection graph). We will also say that label $\ell$ is \emph{strong} if $|L_{\ell}| \geq 3$, otherwise it is \emph{weak}. For a given closed vertex-label sequence $\sigma$, and any integer $\lambda \in [|\sigma|]$, we will say that $\sigma$ is \emph{$\lambda$-strong} if $|L_{\ell_i}| \geq 3$, for $\lambda$ indices $i \in [|\sigma|]$. The structural Lemma below is useful for our analysis (see Section \ref{sec:lemma1proof} of the Appendix for the proof).\footnote{We conjecture that the structural property of Lemma \ref{lem:disjointcycles} also holds if we replace $0$-strong with $\lambda$-strong, for any constant $\lambda$, but this stronger version is not necessary for our analysis.}

\begin{lemma} \label{lem:disjointcycles}
Let $G(V,E, \mathbf{R}^T \mathbf{R})$ be a random instance of the $\overline{{\cal G}}_{n, m, p}$ model, with $m=n$, and $p = \frac{c}{n}$, for some constant $c>0$. With high probability over the choices of $\mathbf{R}$, 0-strong closed vertex-label sequences in $G$ do not have labels in common. 
\end{lemma}

The following definition is essential for the presentation of our algorithm.

\begin{definition} \label{def:codd}
Given a weighted intersection graph $G=G(V,E, \mathbf{R}^T \mathbf{R})$ and a subgraph $G^{(b)} \subseteq G$, let ${\cal C}_{odd}(G^{(b)})$ be the set of odd length closed vertex-label sequences $\sigma:= v_1, \ell_1, v_2,  \\ \ell_2, \cdots, v_k, \ell_{k}, v_{k+1}=v_1$ that additionally satisfy the following:

\begin{description}
\item[(a)] $\sigma$ has distinct vertices (except the first and the last) and distinct labels.
\item[(b)] $v_i$ is connected to $v_{i+1}$ in $G^{(b)}$, for all $i \in [|\sigma|]$.
\item[(c)] $\sigma$ is $\lambda$-strong, for some $\lambda > 0$. 
\end{description}

\end{definition}

Algorithm \ref{alg:WeakBipartization} initially replaces each clique $K^{(\ell)}$ by a random maximal matching $M^{(\ell)}$, and thus gets a subgraph $G^{(b)} \subseteq G$. If ${\cal C}_{odd}(G^{(b)})$ is not empty, then the algorithm selects $\sigma \in {\cal C}_{odd}(G^{(b)})$ and a strong label $\ell \in \sigma$, and then replaces $M^{(\ell)}$ in $G^{(b)}$ by a new random matching of $K^{(\ell)}$. The algorithm repeats until all odd cycles are destroyed (or runs forever trying to do so).



\begin{algorithm}[H] \label{alg:WeakBipartization}
\KwIn{Weighted intersection graph $G = \cup^+_{\ell \in [m]} K^{(\ell)}$} 
\KwOut{A subgraph of $G^{(b)}$ that has only $0$-strong odd cycles}

\For{each $\ell \in [m]$}{
Let $M^{(\ell)}$ be a random maximal matching of $K^{(\ell)}$\;
}

Set $G^{(b)} = \cup^+_{\ell \in [m]} M^{(\ell)}$ \;

\While{${\cal C}_{odd}(G^{(b)}) \neq \emptyset$}{
Let $\sigma \in {\cal C}_{odd}(G^{(b)})$ and $\ell$ a label in $\sigma$ with $|L_{\ell}| \geq 3$\;
Replace the part of $G^{(b)}$ corresponding to $\ell$ by a new random maximal matching $M^{(\ell)}$\;
}

\Return{$G^{(b)}$\;}

\caption{Intersection Graph Weak Bipartization}
\end{algorithm}



The following results are the main technical tools that justify the use of the Weak Bipartization Algorithm for \textsc{Weighted Max Cut}. The proof details for Lemma \ref{lem:only0oddcycles} can be found in Section \ref{sec:only0oddcycles} of the Appendix. 

\begin{lemma} \label{lem:only0oddcycles}
If ${\cal C}_{odd}(G^{(b)})$ is empty, then $G^{(b)}$ may only have 0-strong odd cycles.
\end{lemma}




\begin{theorem} \label{theorem:discrepancytomaxcut}
Let $G(V,E, \mathbf{R}^T \mathbf{R})$ be a random instance of the $\overline{{\cal G}}_{n, m, p}$ model, with $n=m$ and $p = \frac{c}{n}$, where $c>0$ is a constant, and let $\mathbf{R}$ be its representation matrix. Let also $\Sigma$ be a set system with incidence matrix $\mathbf{R}$. With high probability over the choices of $\mathbf{R}$, if Algorithm \ref{alg:WeakBipartization} for weak bipartization terminates on input $G$, its output can be used to construct a 2-coloring $\mathbf{x}^{(\text{disc})} \in \arg \min_{\mathbf{x} \in \{\pm 1\}^n} \text{disc}(\Sigma, \mathbf{x})$, which also gives a maximum cut in $G$, i.e. $\mathbf{x}^{(\text{disc})} \in \arg \max_{\mathbf{x} \in \{\pm 1\}^n} \text{Cut}(G, \mathbf{x})$.
\end{theorem} 
\begin{proof} By construction, the output of Algorithm \ref{alg:WeakBipartization}, namely $G^{(b)}$, has only 0-strong odd cycles. Furthermore, by Lemma \ref{lem:disjointcycles} these cycles correspond to vertex-label sequencies that are label-disjoint. Let $H$ denote the subgraph of $G^{(b)}$ in which we have destroyed all 0-strong odd cycles by deleting a single (arbitrary) edge $e_C$ from each 0-strong odd cycle $C$ (keeping all other edges intact), and notice that $e_C$ corresponds to a weak label. In particular, $H$ is a bipartite multi-graph and thus its vertices can be partitioned into two independent sets $A, B$ constructed as follows: In each connected component of $H$, start with an arbitrary vertex $v$ and include in $A$ (resp. in $B$) the set of vertices reachable from $v$ that are at an even (resp. odd) distance from $v$. Since $H$ is bipartite, it does not have odd cycles, and thus this construction is well-defined, i.e. no vertex can be placed in both $A$ and $B$. 

We now define $\mathbf{x}^{(disc)}$ by setting $x^{(disc)}_i = +1$ if $i \in A$ and $x^{(disc)}_i = +1$ if $i \in B$. Let ${\cal M}_0$ denote the set of weak labels corresponding to the edges removed from $G^{(b)}$ in the construction of $H$. We first note that, for each $\ell_C \in {\cal M}_0$ corresponding to the removal of an edge $e_C$, we have $\left|\sum_{i \in L_{\ell_C}} x^{(disc)}_i \right|=2$. Indeed, since $e_C$ belongs to an odd cycle in $G^{(b)}$, its endpoints are at even distance in $H$, which means that either they both belong to $A$ or they both belong to $B$. Therefore, their corresponding entries of $\mathbf{x}^{(disc)}$ have the same sign, and so (taking into account that the endpoints of $e_C$ are the only vertices in $L_{\ell_C}$), we have $\left|\sum_{i \in L_{\ell_C}} x^{(disc)}_i \right|=2$. Second, we show that, for all the other labels $\ell \in [m] \backslash {\cal M}_0$, $\left|\sum_{i \in L_{\ell}} x^{(disc)}_i \right|$ will be equal to 1 if $|L_{\ell}|$ is odd and 0 otherwise. For any label $\ell \in [m] \backslash {\cal M}_0$, let $M^{(\ell)}$ denote the part of $G^{(b)}$ corresponding to a maximal matching of $K^{(\ell)}$, and note that all edges of $M^{(\ell)}$ are contained in $H$. Since $H$ is bipartite, no edge in $M^{(\ell)}$ can have both its endpoints in either $A$ or $B$. Therefore, by construction, the contribution of entries of $\mathbf{x}^{(disc)}$ corresponding to endpoints of edges in $M^{(\ell)}$ to the sum $\sum_{i \in L_{\ell}} x^{(disc)}_i$ is 0. In particular, if $|L_{\ell}|$ is even, then $M^{(\ell)}$ is a perfect matching and $\left|\sum_{i \in L_{\ell}} x^{(disc)}_i \right| = 0$, otherwise (i.e. if $|L_{\ell}|$ is odd) there is a single vertex not matched in $M^{(\ell)}$ and $\left|\sum_{i \in L_{\ell}} x^{(disc)}_i \right| = 1$.  

To complete the proof of the theorem, we need to show that $\text{Cut}(G, \mathbf{x}^{(disc)})$ is maximum. By Corollary \ref{corollary:maxcutRIGformula}, this is equivalent to proving that $\|\mathbf{R} \mathbf{x}^{(disc)}\| \leq \|\mathbf{R} \mathbf{x}\|$ for all $\mathbf{x} \in \{-1,+1\}^n$. Suppose that there is some $\mathbf{x}^{(min)} \in \{-1,+1\}^n$ such that $\|\mathbf{R} \mathbf{x}^{(disc)}\| > \|\mathbf{R} \mathbf{x}^{(min)}\|$. As mentioned above, for all $\ell \in [m] \backslash {\cal M}_0$, we have $[\mathbf{R} \mathbf{x}^{(disc)}]_{\ell} \leq 1$, and so $[\mathbf{R} \mathbf{x}^{(disc)}]_{\ell} \leq [\mathbf{R} \mathbf{x}^{(min)}]_{\ell}$. Therefore, the only labels where $\mathbf{x}^{(min)}$ could do better are those corresponding to edges $e_C$ that are removed from $G^{(b)}$ in the construction of $H$, i.e. $\ell_C \in {\cal M}_0$, for which we have $[\mathbf{R} \mathbf{x}^{(disc)}]_{\ell_C} =2$. However, any such edge $e_C$ belongs to an odd cycle $C$, and thus any 2-coloring of the vertices of $C$ will force at least one of the 0-strong labels corresponding to edges of $C$ to be monochromatic. Taking into account the fact that, by Lemma \ref{lem:disjointcycles}, with high probability over the choices of $\mathbf{R}$, all 0-strong odd cycles correspond to vertex-label sequences that are label-disjoint, we conclude that $\|\mathbf{R} \mathbf{x}^{(disc)}\| \leq \|\mathbf{R} \mathbf{x}^{(min)}\|$, which completes the proof.
\end{proof}

The fact that Theorem \ref{theorem:discrepancytomaxcut} is not an immediate consequence of Corollary \ref{thm:MaxCutvsDiscrepancy} follows from the observation that a random set system with incidence matrix $\mathbf{R}$ has discrepancy larger than 1 with (at least) constant probability when $m=n$ and $p = \frac{c}{n}$. Indeed, by a straightforward counting argument, we can see that the expected number of 0-strong odd cycles is at least constant. Furthermore, in any 2-coloring of the vertices at least one of the weak labels forming edges in a 0-strong odd cycle will be monochromatic. Therefore, with at least constant probability, for any $\mathbf{x} \in \{-1,+1\}^n$, there exists a weak label $\ell$, such that $x_i x_j=1$, for both $i, j \in L_{\ell}$, implying that $\text{disc}(L_{\ell})=2$. 

We close this section by a result indicating that the conditional statement of Theorem \ref{theorem:discrepancytomaxcut} is not void, namely there is a range of values for $c$ where the Weak Bipartization Algorithm terminates in polynomial time. 

\begin{theorem} \label{theorem:c<1termination}
Let $G(V,E, \mathbf{R}^T \mathbf{R})$ be a random instance of the $\overline{{\cal G}}_{n, m, p}$ model, with $n=m$ and $p = \frac{c}{n}$, where $0<c<1$ is a constant, and let $\mathbf{R}$ be its representation matrix. With high probability over the choices of $\mathbf{R}$, Algorithm \ref{alg:WeakBipartization} for weak bipartization terminates on input $G$ in $O\left((n+\sum_{\ell \in [m]} |L_{\ell}|) \cdot \log{n} \right)$ polynomial time.
\end{theorem}

The proof of the above theorem uses the following structural Lemma regarding the expected number of closed vertex label sequences. 

\begin{lemma} \label{lemma:kcyckles}
Let $G(V,E, \mathbf{R}^T \mathbf{R})$ be a random instance of the $\overline{{\cal G}}_{n, m, p}$ model. Let also $C_k$ denote the number of distinct closed vertex-label sequences of size $k$ in $G$. Then 
\begin{equation} \label{eq:C_k}
\mathbb{E}[C_k] = \frac{1}{k} \frac{n!}{(n-k)!} \frac{m!}{(m-k)!} p^{2k}.   
\end{equation}
In particular, when $m=n \to \infty$, $p = \frac{c}{n}, c>0$, and $k \geq 3$, we have $\mathbb{E}[C_k] \leq \frac{e}{2\pi} c^{2k}$. 
\end{lemma}
\begin{proof}
Notice that there are $\frac{1}{k} \frac{n!}{(n-k)!}$ ways to arrange $k$ out of $n$ vertices in a cycle. Furthermore, in each such arrangement, there are $\frac{m!}{(m-k)!}$ ways to place $k$ out of $m$ labels so that there is exactly one label between each pair of vertices. Since labels in any given arrangement must be selected by both its adjacent vertices, (\ref{eq:C_k}) follows by linearity of expectation. 

Setting $m=n$ and $p = \frac{c}{n}$, and using the inequalities $\sqrt{2 \pi} n^{n+\frac{1}{2}}e^{-n} \leq n! \leq e n^{n+\frac{1}{2}}e^{-n}$,
\begin{eqnarray}
\mathbb{E}[C_k] &=& \frac{1}{k} \left(\frac{n!}{(n-k)!} \right)^2 \left( \frac{c}{n}\right)^{2k} \nonumber \\
& \leq & \frac{1}{k} \frac{e^2 n^{2n+1} e^{-2n}}{2\pi (n-k)^{2n-2k+1} e^{2k-2n}} \left(\frac{c}{n}\right)^{2k} = \frac{1}{k} \frac{e^2}{2\pi} \left( \frac{n}{n-k}\right)^{2n-2k+1} \left( \frac{c}{e} \right)^{2k} \nonumber \\
& \leq & \frac{e^2}{2\pi} \frac{n}{k (n-k)} e^{\frac{k}{n-k} (2n-2k)} \left( \frac{c}{e} \right)^{2k} = \frac{e^2}{2\pi} \frac{n}{k (n-k)} c^{2k}. \nonumber
\end{eqnarray}
When $n$ goes to $\infty$ and $k \geq 3$, then the above is at most $\frac{e}{2\pi} c^{2k}$ as needed.
\end{proof}

We are now ready for the proof of the Theorem.

\begin{proof}[Proof of Theorem \ref{theorem:c<1termination}]
We will prove that, when $m=n \to \infty$, $p = \frac{c}{n}, c<1$, and $k \geq 3$, with high probability, there are no closed vertex-label sequences that have labels in common. To this end, recalling Definition \ref{def:codd} for ${\cal C}_{odd}(G^{(b)})$, we provide upper bounds on the following events: $A \stackrel{\text{def}}{=} \{\exists k \geq \log{n}: C_k \geq 1\}$, $B \stackrel{\text{def}}{=} \{|{\cal C}_{odd}(G^{(b)})| \geq \log{n}\}$ and $C \stackrel{\text{def}}{=} \{\exists \sigma \neq \sigma' \in {\cal C}_{odd}(G^{(b)}): \exists \ell \in \sigma, \ell \in \sigma'\}$.

By the union bound, Markov's inequality and Lemma \ref{lemma:kcyckles}, we get that, whp all closed vertex-label sequences have less than $\log{n}$ labels:

\begin{equation} \label{ineq:nolargecycles}
\Pr\left(A \right) \leq \sum_{k \geq \log{n}} \mathbb{E}[C_k] \leq \sum_{k \geq \log{n}} \frac{e}{2\pi} c^{2k} = \frac{e}{2 \pi} \frac{c^{2 \log{n}}}{1-c^2} = O\left( c^{2 \log{n}} \right) = o(1),
\end{equation}
where the last equality follows since $c<1$ is a constant. Furthermore, by Markov's inequality and Lemma \ref{lemma:kcyckles}, and noting that any closed vertex-label sequence in ${\cal C}_{odd}(G^{(b)})$ must have at least $k \geq 3$ labels, we get that, whp there less than $\log{n}$ closed vertex-label sequences in ${\cal C}_{odd}(G^{(b)})$:
\begin{equation} \label{ineq:notmanycycles}
\Pr\left(B \right) \leq \frac{1}{\log{n}} \sum_{k \geq 3} \mathbb{E}[C_k] \leq \frac{1}{\log{n}} \sum_{k \geq 3} \frac{e}{2\pi} c^{2k} = \frac{1}{\log{n}} \frac{e}{2 \pi} \frac{c^{6}}{1-c^2} = O\left( \frac{1}{\log{n}} \right).
\end{equation}

To bound $\Pr(C)$, fix a closed vertex-label sequence $\sigma$, and let $|\sigma| \geq 3$ be the number of its labels. Notice that, the probability that there is another closed vertex-label sequence that has labels in common with $\sigma$ implies the existence of a vertex-label sequence $\breve{\sigma}$ that starts with either a vertex or a label from $\sigma$, ends with either a vertex or a label from $\sigma$, and has at least one label or at least one vertex that does not belong to $\sigma$. Let $|\breve{\sigma}|$ denote the number of labels of $\breve{\sigma}$ that do not belong to $\sigma$. Then the number of different vertex-label sequences $\breve{\sigma}$ that start and end in labels from $\sigma$ is at most $|\sigma|^2 n^{|\breve{\sigma}|+1} m^{|\breve{\sigma}|}$; indeed $\breve{\sigma}$ in this case has $|\breve{\sigma}|$ labels and $|\breve{\sigma}|+1$ vertices that do not belong to $\sigma$. Therefore, by independence, each such sequence $\breve{\sigma}$ has probability $p^{2|\breve{\sigma}|+2}$ to appear. Similarly, the number of different vertex-label sequences $\breve{\sigma}$ that start and end in vertices from $\sigma$ is at most $|\sigma|^2 n^{|\breve{\sigma}|-1} m^{|\breve{\sigma}|}$ and each one has probability $p^{2|\breve{\sigma}|}$ to appear. Finally, the number of different vertex-label sequences $\breve{\sigma}$ that start in a vertex from $\sigma$ and end in a label from $\sigma$ (notice that this also covers the case where $\breve{\sigma}$ starts in a label from $\sigma$ and ends in a vertex from $\sigma$) is at most $|\sigma|^2 n^{|\breve{\sigma}|} m^{|\breve{\sigma}|}$ and each one has probability $p^{2|\breve{\sigma}|+1}$ to appear. Overall, for a given sequence $\sigma$, the expected number of sequences $\breve{\sigma}$ described above that additionally satisfies $|\breve{\sigma}| < \log{n}$, is at most 

\begin{equation} \label{eq:boundedsigmas}
\sum_{k=0}^{\log{n}-1} |\sigma|^2 n^{k+1} m^{k} p^{2k+2} + \sum_{k=1}^{\log{n}-1} |\sigma|^2 n^{k-1} m^{k} p^{2k} + \sum_{k=1}^{\log{n}-1} |\sigma|^2 n^{k} m^{k} p^{2k+1} \leq c |\sigma|^2 \frac{\log{n}}{n},
\end{equation}
where in the last inequality we used the fact that $m=n, p = \frac{c}{n}$ and $c<1$. Since the existence of a sequence $\breve{\sigma}$ for $\sigma$ that additionally satisfies $|\breve{\sigma}| \geq \log{n}$ implies event $A$, and on other hand the existence of more than $\log{n}$ different sequences $\sigma \in |{\cal C}_{odd}(G^{(b)})|$ implies event $B$, by Markov's inequality and (\ref{eq:boundedsigmas}), we get 
\begin{equation*}
\Pr(C) \leq \Pr(A) + \Pr(B) + c \frac{(\log{n})^4}{n} =  O\left( c^{2 \log{n}} \right) + O\left( \frac{1}{\log{n}} \right) + O\left(\frac{(\log{n})^4}{n} \right) = O\left( \frac{1}{\log{n}} \right).
\end{equation*} 
We have thus proved that, with high probability over the choices of $\mathbf{R}$, closed vertex-label sequences in ${\cal C}_{odd}(G^{(b)})$ are label disjoint, as needed. 

In view of this, the proof of the Theorem follows by noting that, since closed vertex label sequences in ${\cal C}_{odd}(G^{(b)})$ are label disjoint, steps 5 and 6 within the while loop of the Weak Bipartization Algorithm will be executed exactly once for each sequence in ${\cal C}_{odd}(G^{(b)})$, where $G^{(b)}$ is defined in step 3 of the algorithm; indeed, once a closed vertex label sequence $\sigma \in {\cal C}_{odd}(G^{(b)})$ is destroyed in step 6, no new closed vertex label sequence is created. In fact, once $\sigma$ is destroyed we can remove the corresponding labels and edges from $G^{(b)}$, as these will no longer belong to other closed vertex label sequences. Furthermore, to find a closed vertex label sequences in ${\cal C}_{odd}(G^{(b)})$, it suffices to find an odd cycle in $G^{(b)}$, which can be done  by running DFS, requiring $O(n+\sum_{\ell \in [m]} |L_{\ell}|)$ time, because $G^{(b)}$ has at most $\sum_{\ell \in [m]} |L_{\ell}|$ edges. Finally, by (\ref{ineq:notmanycycles}), we have $|{\cal C}_{odd}(G^{(b)})| < \log{n}$ with high probability, and so the running time of the Weak Bipartization Algorithm is $O((n+\sum_{\ell \in [m]} |L_{\ell}|) \log{n})$, which concludes the proof of Theorem \ref{theorem:c<1termination}.
\end{proof}

\section{Discussion and some open problems}

In this paper, we introduced the model of weighted random intersection graphs and we studied the average case analysis of \textsc{Weighted Max Cut} through the prism of discrepancy of random set systems. In particular, in the first part of the paper, we proved concentration of the weight of a maximum cut of $G(V, E, \mathbf{R}^T \mathbf{R})$ around its expected value, and we used it to show that, with high probability, the weight of a random cut is asymptotically equal to the maximum cut weight of the input graph, when $m = n^{\alpha}, \alpha<1$. On the other hand, in the case where the number of labels is equal to the number of vertices (i.e. $m=n$), we proved that a majority algorithm gives a cut with weight that is larger than the weight of a random cut by at least a constant factor, when $p = \frac{c}{n}$ and $c$ is large. 

In the second part of the paper, we highlighted a connection between \textsc{Weighted Max Cut} of sparse weighted random intersection graphs and \textsc{Discrepancy} of sparse random set systems, formalized through our Weak Bipartization Algorithm and its analysis. We demonstrated how our proposed framework can be used to find optimal solutions for these problems, with high probability, in special cases of sparse inputs ($m=n, p=\frac{c}{n}, c<1$).

One of the main problems left open in our work concerns the termination of our Weak Bipartization Algorithm for large values of $c$. We conjecture the following:

\begin{conjecture} Let $G(V,E, \mathbf{R}^T \mathbf{R})$ be a random instance of the $\overline{{\cal G}}_{n, m, p}$ model, with $m=n$, and $p = \frac{c}{n}$, for some constant $c \geq 1$. With high probability over the choices of $\mathbf{R}$, on input $G$, Algorithm \ref{alg:WeakBipartization} for weak bipartization terminates in polynomial time. 
\end{conjecture}

We also leave the problem of determining whether Algorithm \ref{alg:WeakBipartization} terminates in polynomial time, in the case $m=n$ and $p = \omega(1/n)$, as an open question for future research.

Towards strengthening the connection between \textsc{Weighted Max Cut} under the $\overline{{\cal G}}_{n, m, p}$ model, and \textsc{Discrepancy} in random set systems, we conjecture the following:




\begin{conjecture}
Let $G(V,E, \mathbf{R}^T \mathbf{R})$ be a random instance of the $\overline{{\cal G}}_{n, m, p}$ model, with $m=n^{\alpha}, \alpha \leq 1$ and $mp^2 = O(1)$, and let $\mathbf{R}$ be its representation matrix. Let also $\Sigma$ be a set system with incidence matrix $\mathbf{R}$. Then, with high probability over the choices of $\mathbf{R}$, there exists $\mathbf{x}^{\text{disc}} \in \arg \min_{\mathbf{x} \in \{-1, +1\}^n} \text{disc}(\Sigma, \mathbf{x})$, such that $ \texttt{Cut}(G, \mathbf{x}^{\text{disc}})$ is asymptotically equal to $\texttt{Max-Cut}(G)$. 
\end{conjecture}

\newpage




\newpage

\appendix


\section{Proof of Corollary \ref{corollary:maxcutRIGformula}} \label{sec:corollary1proof}

We first prove the following Lemma, by straightforward calculation from equation (\ref{eq:maxcut}):
\begin{lemma} \label{lem:maxcutmatrix}
Let $G(V, E, \mathbf{W})$ be a weighted graph such that $\mathbf{W}$ is symmetric and $\mathbf{W}_{i,j} = 0$ if $\{i,j\} \notin E$. Then 
\begin{equation} \label{eq:maxcutmatrix}
\texttt{Max-Cut}(G) = \frac{1}{4} \left(\sum_{i,j \in [n]^2} \mathbf{W}_{i,j} - \min_{\mathbf{x} \in \{-1, +1\}^n} \mathbf{x}^T \mathbf{W} \mathbf{x} \right).
\end{equation}
\end{lemma}
\begin{proof} For any $\mathbf{x} \in \{-1,+1\}^n$, we write
\begin{eqnarray*}
\sum_{i,j \in [n]^2} \mathbf{W}_{i,j} - \mathbf{x}^T \mathbf{W} \mathbf{x} & = & \sum_{i,j \in [n]^2} \mathbf{W}_{i,j} - \sum_{i,j \in [n]^2} \mathbf{W}_{i,j} x_i x_j \\
& = & \frac{1}{2} \sum_{i,j \in [n]^2} \mathbf{W}_{i,j} \left( x_i^2+x_j^2 - 2x_i x_j\right) \\
& = & \frac{1}{2} \sum_{i,j \in [n]^2} \mathbf{W}_{i,j} \left( x_i-x_j\right)^2 \\
& = & \sum_{\{i,j\} \in E} \mathbf{W}_{i,j} \left( x_i-x_j\right)^2. 
\end{eqnarray*}
By (\ref{eq:maxcut}), this completes the proof.
\end{proof}

\begin{proof}[Proof of Corollary \ref{corollary:maxcutRIGformula}]
Notice that diagonal entries of the weight matrix in (\ref{eq:maxcutmatrix}) cancel out, and so, for any $\mathbf{x} \in \{-1, +1\}^n$, we have 
\begin{equation*}
\sum_{i,j \in [n]^2} \left[\mathbf{R}^T \mathbf{R} \right]_{i,j} - \left\| \mathbf{R} \mathbf{x} \right\|^2 = \sum_{i\neq j, i,j \in [n]^2} \left[\mathbf{R}^T \mathbf{R} \right]_{i,j} - \sum_{i\neq j, i,j \in [n]^2} \left[\mathbf{R}^T \mathbf{R} \right]_{i,j} x_i x_j.
\end{equation*}
Taking expectations with respect to $\mathbf{x}$, the contribution of the second sum in the above expression equals 0, which completes the proof.
\end{proof}

\section{Proof of Corollary \ref{thm:MaxCutvsDiscrepancy}} \label{sec:previoustheorem}

\begin{proof} Since $\text{disc}(\Sigma, \mathbf{x}^*) \leq 1$, then each component of $\mathbf{R}\mathbf{x}^*$ is either 0 or 1, for any $\mathbf{x}^* \in \{-1, +1\}^n$. In particular, for any $\ell \in [m]$, $\left[\mathbf{R}\mathbf{x}^*\right]_{\ell}$ is 0 if the number of ones in the $\ell$-th row is even and it is equal to 1 otherwise. This is the best one can hope for, since sets with an odd number of elements cannot have discrepancy less than 1. Therefore, $\|\mathbf{R} \mathbf{x}^*\|$ is also the minimum possible. In particular, this implies that, in the case $\text{disc}(\Sigma, \mathbf{x}^*) \leq 1$, any 2-coloring that achieves minimum discrepancy gives a bipartition that corresponds to a maximum cut and vice versa.
\end{proof}

\section{Proof of Theorem \ref{theorem:randomcut}} \label{sec:theorem3proof}

\begin{proof} Let $G=G(V,E, \mathbf{R}^T \mathbf{R})$ be a weighted random intersection graph. By equation (\ref{eq:cutRIGformula}) of Corollary \ref{corollary:maxcutRIGformula}, for any $\mathbf{x} \in \{-1, +1\}^n$, we have:
\begin{equation*}
\texttt{Cut}(G, \mathbf{x}) = \frac{1}{4} \left(\sum_{i,j \in [n]} \left[\mathbf{R}^T \mathbf{R} \right]_{i,j} - \|\mathbf{R} \mathbf{x}\|^2 \right).
\end{equation*}
Taking expectations with respect to random $\mathbf{x}$ and $\mathbf{R}$, we get
\begin{eqnarray}
\mathbb{E}_{\mathbf{x}, \mathbf{R}}[\texttt{Cut}(G, \mathbf{x})] & = & \frac{1}{4} \cdot \mathbb{E}_{\mathbf{R}}\left[ \sum_{i,j \in [n]} \left[\mathbf{R}^T \mathbf{R} \right]_{i,j} - \sum_{i \in [n]} \left[\mathbf{R}^T \mathbf{R} \right]_{i,i} \right] \nonumber \\
& = & \frac{1}{4} \cdot \mathbb{E}_{\mathbf{R}}\left[ \sum_{i\neq j, i,j \in [n]} \left[\mathbf{R}^T \mathbf{R} \right]_{i,j} \right] = \frac{1}{4} n(n-1)mp^2. \label{eq:randomcut}
\end{eqnarray}
To prove the Theorem, we will show that, with high probability over random $\mathbf{x}$ and $\mathbf{R}$, we have $\|\mathbf{R} \mathbf{x}\|^2 = o\left( \mathbb{E}_{\mathbf{R}}\left[ \frac{1}{4} \sum_{i\neq j, i,j \in [n]} \left[\mathbf{R}^T \mathbf{R} \right]_{i,j} \right] \right) = o(n^2mp^2)$, in which case the theorem follows by concentration of $\texttt{Max-Cut}(G)$ around its expected value (Theorem \ref{theorem:MaxCutconcentration}), and the fact that $\texttt{Max-Cut}(G) \geq \frac{1}{4} \sum_{i\neq j, i,j \in [n]} \left[\mathbf{R}^T \mathbf{R} \right]_{i,j}$. 


To this end, fix $\ell \in [m]$ and consider the random variable counting the number of ones in the $\ell$-th row of $\mathbf{R}$, namely $Y_{\ell} = \sum_{i \in [n]} \mathbf{R}_{\ell, i}$. By the multiplicative Chernoff bound, for any $\delta >0$, 
\begin{equation*}
\Pr(Y_{\ell}> (1+\delta) np ) \leq \left( \frac{e^{\delta}}{(1+\delta)^{1+\delta}} \right)^{np}.
\end{equation*}	
Since $np \geq C_1 \sqrt{\frac{n}{m}} = C_1 n^{\frac{1-\alpha}{2}}$, taking any $\delta \geq 2$, we get
\begin{equation} \label{eq:yellbound}
\Pr(Y_{\ell}> 3 np ) \leq \left( \frac{e^{2}}{27} \right)^{np} = o\left( \frac{1}{m} \right).
\end{equation}	
Therefore, by the union bound, 
\begin{equation} \label{eq:unionbound}
\Pr(\exists \ell \in [m]: Y_{\ell}> 3 np ) = o(1),
\end{equation}
implying that, all rows of $\mathbf{R}$ have at most $3np$ non-zero elements with high probability.

Fix now $\ell$ and consider the random variable corresponding to the $\ell$-th entry of $\mathbf{R} \mathbf{x}$, namely $Z_{\ell} = \sum_{i \in [n]} \mathbf{R}_{\ell, i} x_i$. In particular, given $Y_{\ell}$, notice that $Z_{\ell}$ is equal to the sum of $Y_{\ell}$ independent random variables $x_i \in \{-1, +1\}$, for $i$ such that $\mathbf{R}_{\ell, i}=1$. Therefore, since $\mathbb{E}_{\mathbf{x}}[Z_{\ell}] = \mathbb{E}_{\mathbf{x}}[Z_{\ell} | Y_{\ell}]=0$, by Hoeffding's inequality, for any $\lambda \geq 0$,
\begin{equation*}
\Pr(|Z_{\ell}|>\lambda | Y_{\ell}) \leq e^{-\frac{\lambda^2}{2Y_{\ell}}}.
\end{equation*}
Therefore, by the union bound, and taking $\lambda \geq \sqrt{6 np \ln{n}}$,

\begin{equation} \label{eq:overallbound}
\Pr(|Z_{\ell}|>\lambda) \leq \Pr(\exists \ell \in [m]: Y_{\ell}> 3 np ) + m e^{-\frac{\lambda^2}{6np}} = o(1)+ \frac{m}{n} = o(1),
\end{equation}
implying that all entries of $\mathbf{R} \mathbf{x}$ have absolute value at most $\sqrt{6 np \ln{n}}$ with high probability over the choices of $\mathbf{x}$ and $\mathbf{R}$. Consequently, with high probability over the choices of $\mathbf{x}$ and $\mathbf{R}$, we have $\| \mathbf{R} \mathbf{x}\|^2 = 6mnp \ln{n}$, which is $o(n^2mp^2)$, since $np = \omega(\ln{n})$ in the range of parameters considered in this theorem. This completes the proof.
\end{proof}

We note that the same analysis also holds when $n=m$ and $p$ is sufficiently large (e.g. $p = \omega(\frac{\ln{n}}{n})$). In particular, similar probability bounds hold in equations (\ref{eq:yellbound}), (\ref{eq:unionbound}) and (\ref{eq:overallbound}), for the same choices of $\delta \geq 2$ and $\lambda \geq \sqrt{6 np \ln{n}}$, implying that $\| \mathbf{R} \mathbf{x}\|^2 = 6mnp \ln{n} = o(n^2mp^2)$ with high probability. 


\section{Proof of Theorem \ref{thm:majoritycutanalysis}} \label{sec:majoritycutanalysis}




\begin{proof} Let $G(V,E, \mathbf{R}^T \mathbf{R})$ (i.e. the input to the Majority Cut Algorithm \ref{alg:majoritycut}) be a random instance of the $\overline{{\cal G}}_{n, m, p}$ model, with $m=n$, and $p = \frac{c}{n}$, for some large enough constant $c$. For $t \in [n]$, let $M_t$ denote the constructed cut size just after the consideration of a vertex $v_t$, for some $t \geq \epsilon n+1$. In particular, by equation (\ref{eq:maxcutRIGformula}) for $n=t$, and since the values $x_1, \ldots, x_{t-1}$ are already decided in previous steps, we have
\begin{eqnarray}
M_t & = & \frac{1}{4} \left(\sum_{i,j \in [t]^2} \left[\mathbf{R}^T \mathbf{R} \right]_{i,j} - \min_{x_t \in \{-1, +1\}}  \left\| \mathbf{R}_{[m], [t]} \mathbf{x}_{[t]} \right\|^2 \right) \label{eq:Mt} 
\end{eqnarray}

The first of the above terms is
\begin{equation} \label{eq:1stterm}
\frac{1}{4} \sum_{i,j \in [t]^2} \left[\mathbf{R}^T \mathbf{R} \right]_{i,j} = \frac{1}{4} \left( \sum_{i,j \in [t-1]^2} \left[\mathbf{R}^T \mathbf{R} \right]_{i,j} + 2 \sum_{i \in [t-1]} \left[\mathbf{R}^T \mathbf{R} \right]_{i,t} + \left[\mathbf{R}^T \mathbf{R} \right]_{t,t}\right) 
\end{equation}
and the second term is
\begin{eqnarray}
&& -\frac{1}{4} \min_{x_t \in \{-1, +1\}}  \left\| \mathbf{R}_{[m], [t]} \mathbf{x}_{[t]} \right\|^2 \nonumber \\
&& \quad = -\frac{1}{4} \min_{x_t \in \{-1, +1\}}  \left\| \mathbf{R}_{[m],t} x_t + \sum_{i \in [t-1]} \mathbf{R}_{[m],i} x_i \right\|^2 \nonumber \\
&& \quad = -\frac{1}{4} \min_{x_t \in \{-1, +1\}} \sum_{i,j \in [t]^2} \left[\mathbf{R}^T \mathbf{R} \right]_{i,j} x_i x_j \nonumber \\
&& \quad = -\frac{1}{4} \left( \sum_{i,j \in [t-1]^2} \left[\mathbf{R}^T \mathbf{R} \right]_{i,j} x_i x_j + 2 \min_{x_t \in \{-1, +1\}}{ \sum_{i \in [t-1]} \left[\mathbf{R}^T \mathbf{R} \right]_{i,t} x_i x_t} + \left[\mathbf{R}^T \mathbf{R} \right]_{t,t} \right) \label{eq:2ndterm}
\end{eqnarray}
By (\ref{eq:Mt}), (\ref{eq:1stterm}) and (\ref{eq:2ndterm}), we have

\begin{eqnarray}
M_t & = & M_{t-1} + \frac{1}{2} \sum_{i \in [t-1]} \left[\mathbf{R}^T \mathbf{R} \right]_{i,t} - \frac{1}{2} \min_{x_t \in \{-1, +1\}}{ \sum_{i \in [t-1]} \left[\mathbf{R}^T \mathbf{R} \right]_{i,t} x_i x_t} \nonumber \\
& = & M_{t-1} + \frac{1}{2} \sum_{i \in [t-1]} \left[\mathbf{R}^T \mathbf{R} \right]_{i,t} + \frac{1}{2} \left| \sum_{i \in [t-1]} \left[\mathbf{R}^T \mathbf{R} \right]_{i,t} x_i\right|  \label{eq:recursion}
\end{eqnarray}

Define now the random variable 
\begin{equation*}
Z_t = Z_t(\mathbf{x}, \mathbf{R}) = \sum_{i \in [t-1]} \left[\mathbf{R}^T \mathbf{R} \right]_{i,t} x_i = \sum_{\ell \in [m]} \mathbf{R}_{\ell, t} \sum_{i \in [t-1]}  \mathbf{R}_{\ell,i} x_i,
\end{equation*}
so that $M_t = M_{t-1} + \frac{1}{2} \sum_{i \in [t-1]} \left[\mathbf{R}^T \mathbf{R} \right]_{i,t} + \frac{1}{2} \left| Z_t\right|$. Observe that, in the latter recursive equation, the term $\frac{1}{2} \sum_{i \in [t-1]} \left[\mathbf{R}^T \mathbf{R} \right]_{i,t}$ corresponds to the expected increment of the constructed cut if the $t$-vertex chose its color uniformly at random. Therefore, lower bounding the expectation of $\frac{1}{2}\left| Z_t\right|$ will tell us how much better the Majority Algorithm does when considering the $t$-th vertex.

Towards this end, we first note that, given $\mathbf{x}_{[t-1]} = \{x_i, i \in [t-1] \}$, and $\mathbf{R}_{[m], [t-1]}=\{ \mathbf{R}_{\ell, i}, \ell \in [m], i \in [t-1]\}$, $Z_t$ is the sum of $m$ independent random variables, since the Bernoulli random variables $\mathbf{R}_{\ell,t}, \ell \in [m],$ are independent, for any given $t$ (note that the conditioning is essential for independence, otherwise the inner sums in the definition of $Z_t$ would also depend on the $x_i$'s, which are not random when $i$ is large). Furthermore, $\mathbb{E}[Z_t| \mathbf{x}_{[t-1]}, \mathbf{R}_{[m], [t-1]}] = p \sum_{\ell \in [m]} \sum_{i \in [t-1]}  \mathbf{R}_{\ell,i} x_i$ and $\text{Var}(Z_t| \mathbf{x}_{[t-1]}, \mathbf{R}_{[m], [t-1]}) = p(1-p) \sum_{\ell \in [m]} \left(\sum_{i \in [t-1]} \mathbf{R}_{\ell,i} x_i \right)^2$. Given $\mathbf{x}_{[t-1]} = \{x_i, i \in [t-1] \}$, and $\mathbf{R}_{[m], [t-1]}=\{ \mathbf{R}_{\ell, i}, \ell \in [m], i \in [t-1]\}$, define the sets $A^+_t = \{\ell \in [m]: \sum_{i \in [t-1]}  \mathbf{R}_{\ell,i} x_i > 0\}$ and $A^-_t = \{\ell \in [m]: \sum_{i \in [t-1]}  \mathbf{R}_{\ell,i} x_i < 0\}$. In particular, given $\mathbf{x}_{[t-1]} = \{x_i, i \in [t-1] \}$, and $\mathbf{R}_{[m], [t-1]}=\{ \mathbf{R}_{\ell, i}, \ell \in [m], i \in [t-1]\}$, $Z_t$ can be written as
\begin{equation} \label{eq:Z_t}
Z_t = \sum_{\ell \in A^+_t} \mathbf{R}_{\ell, t} \sum_{i \in [t-1]}  \mathbf{R}_{\ell,i} x_i - \sum_{\ell \in A^-_t} \mathbf{R}_{\ell, t} \left|\sum_{i \in [t-1]}  \mathbf{R}_{\ell,i} x_i \right|,
\end{equation}
where $\mathbf{R}_{\ell, t}, \ell \in A^+_t \cup A^-_t$ are independent Bernoulli random variables with success probability $p$.

It is a matter of careful calculation to show that $\mathbb{E}\left[|Z_t| \big| \mathbf{x}_{[t-1]}, \mathbf{R}_{[m], [t-1]}\right]$ is smallest when the conditional expectation $\mathbb{E}\left[Z_t \big| \mathbf{x}_{[t-1]}, \mathbf{R}_{[m], [t-1]}\right]$ is 0, which happens when the sum of positive factors for the Bernoulli random variables in the definition of $Z_t$ is equal to the sum of negative ones, namely $\sum_{\ell \in A^+_t} \sum_{i \in [t-1]}  \mathbf{R}_{\ell,i} x_i = \sum_{\ell \in A^-_t} \left|\sum_{i \in [t-1]}  \mathbf{R}_{\ell,i} x_i \right|$. Furthermore, we note that $\mathbb{E}[|Z_t| \big| \mathbf{x}_{[t-1]}, \mathbf{R}_{[m], [t-1]}]$ does not increase if we replace $\sum_{\ell \in A^+_t} \mathbf{R}_{\ell, t} \sum_{i \in [t-1]}  \mathbf{R}_{\ell,i} x_i$ and $\sum_{\ell \in A^-_t} \mathbf{R}_{\ell, t} \left|\sum_{i \in [t-1]}  \mathbf{R}_{\ell,i} x_i \right|$ in the expression (\ref{eq:Z_t}) for $Z_t$ by independent binomial random variables $Z_t^+ \sim {\cal B}\left( \sum_{\ell \in A^+_t} \sum_{i \in [t-1]}  \mathbf{R}_{\ell,i} x_i, p\right)$ and $Z_t^- \sim {\cal B}\left( \sum_{\ell \in A^-_t} \left|\sum_{i \in [t-1]}  \mathbf{R}_{\ell,i} x_i \right|, p\right)$, respectively.\footnote{This property follows inductively, by noting that, if $X = \sum_{i=1}^k a_i X_i - \sum_{i=k}^N a_i X_i$, and $X'=\sum_{i=1}^{k-1} a_i X_i + (a_k-1)X_k+X'_k - \sum_{i=k}^N a_i X_i$, where $k, N, a_i \in \mathbb{N}^+, i \in [N]$, and $X_i, i \in [N], X'_k$ are independent, identically distributed Bernoulli random variables, then $\mathbb{E}[|X|] \geq \mathbb{E}[|X'|]$. Indeed, notice that, the independence of $X_k, X'_k$ implies that these random variables work against each other (with respect to the absolute value) at least half of the time.}

In view of the above, if $Z^B_t$ is a random variable which, given $\mathbf{x}_{[t-1]} = \{x_i, i \in [t-1] \}$, and $\mathbf{R}_{[m], [t-1]}=\{ \mathbf{R}_{\ell, i}, \ell \in [m], i \in [t-1]\}$, follows the Binomial distribution ${\cal B}\left( N_t, p\right)$, where
\begin{equation} \label{eq:Nt}
N_t\stackrel{\text{def}}{=}\max\left( \sum_{\ell \in A^+_t} \sum_{i \in [t-1]}  \mathbf{R}_{\ell,i} x_i, \sum_{\ell \in A^-_t} \left|\sum_{i \in [t-1]}  \mathbf{R}_{\ell,i} x_i \right| \right),
\end{equation}
then
\begin{equation}
\mathbb{E}[|Z_t| \big| \mathbf{x}_{[t-1]}, \mathbf{R}_{[m], [t-1]}] \geq \mathrm{MD}(Z^B_t), \label{ineq:MD}
\end{equation}
where $\mathrm{MD}(\cdot)$ is the mean absolute difference of (two independent copies of) $Z^B_t$. In particular, $\mathrm{MD}(Z^B_t) = \mathbb{E}[\left|Z^B_t - Z'^B_t \right|]$, where $Z^B_t, Z'^B_t$ are independent random variables following ${\cal B}\left( N_t, p\right)$. Unfortunately, we are aware of no simple closed formula for $\mathrm{MD}(Z^B_t)$, and so we resort to Gaussian approximation through the Berry-Esseen Theorem:

\begin{theorem*}[Berry-Esseen Theorem \cite{S11}]
Let $X_1, X_2, \ldots,$ be independent, identically distributed random variables, with $\mathbb{E}[X_i]=0, \mathbb{E}[X_i^2] = \sigma^2>0$, and $\mathbb{E}[|X_i|^3] = \rho < \infty$. For $N>0$, let $F_N(\cdot)$ be the cumulative distribution function of $\frac{X_1+\cdots+X_N}{\sigma \sqrt{N}}$, and let $\Phi(\cdot)$ be the cumulative distribution function of the standard normal distribution. Then, $\sup_{x \in \mathbb{R}}|F_N(x)-\Phi(x)| \leq \frac{0.4748 \rho}{\sigma^3 \sqrt{N}}$.    
\end{theorem*}

In our case, we write $Z^B_t = \sum_{i=1}^{N_t} Z^B_{t,i}$, $Z'^B_t = \sum_{i=1}^{N_t} Z'^B_{t,i}$, and set $X_i = Z^B_{t,i} - Z'^B_{t,i}$, where $Z^B_{t,i}, Z'^B_{t,i}$ are independent Bernoulli random variables with success probability $p$, for any $i \in [N_t]$. In particular, we have $\mathbb{E}[X_i]=0$, $\mathbb{E}[X_i^2] = \mathbb{E}[|X_i|^3] = 2p(1-p)$. Therefore, by the Berry-Esseen Theorem, given $\mathbf{x}_{[t-1]} = \{x_i, i \in [t-1] \}$, and $\mathbf{R}_{[m], [t-1]}=\{ \mathbf{R}_{\ell, i}, \ell \in [m], i \in [t-1]\}$,the distribution of $Z^B_t - Z'^B_t$ is approximately Normal ${\cal N}(0, 2p(1-p)N_t)$, with approximation error $\frac{0.4748}{\sqrt{2p(1-p) N_t}}$. 

Notice that the latter approximation error bound becomes $o(1)$ if $N_t = \Theta(n), p = \frac{c}{n}$ and $c \to \infty$. Therefore, we next show that, with high probability over the choices of $\mathbf{R}$, $N_t = \Theta(n)$, for any $t \geq \epsilon n+1$, where $\epsilon$ is the constant used in the Majority Algorithm. In particular, even though we cannot control the  variables $x_i \in \{-1,+1\}, i \in [t-1]$, in the definition of $N_t$, we will find a lower bound that holds whp by using the random variable 
\begin{equation*}
Y_t = Y_t(\mathbf{R}, \mathbf{x}) \stackrel{\text{def}}{=} \left| \ell \in [m]: \sum_{i \in [t-1]} \mathbf{R}_{\ell,i} \textrm{ is odd}\right|, 
\end{equation*}
and employing the following inequality 
\begin{equation} \label{ineq:N_tY_t}
N_t \geq \frac{Y_t}{2}.
\end{equation}
Indeed, (\ref{ineq:N_tY_t}) holds because, for any $i \in [t-1]$, if $\sum_{i \in [t-1]} \mathbf{R}_{\ell,i}$ is odd, then $\left|\sum_{i \in [t-1]} \mathbf{R}_{\ell,i} x_i \right| \geq 1$, no matter what value the $x_i$'s have. Therefore, $\sum_{i \in [t-1]} \mathbf{R}_{\ell,i} x_i$ will contribute at least 1 to one of the two terms in the maximum from the right side of (\ref{eq:Nt}), and thus (\ref{ineq:N_tY_t}) follows.

Notice now that, for any fixed $i$ and $t \geq \epsilon n+1$, we have $\Pr(\sum_{i \in [t-1]} \mathbf{R}_{\ell,i} \textrm{ is odd}) = \sum_{j \textrm{ odd}} \binom{t-1}{j} p^j (1-p)^{t-1-j} = \frac{1}{2} \left( 1 - (1-2p)^{t-1}\right) \geq \frac{1}{2} \left( 1 - e^{-2p(t-1)}\right) \geq \frac{1}{2} \left( 1 - e^{-2c\epsilon}\right)$, where in the last inequality we set $p = \frac{c}{n}$. Taking $c \to \infty$, the latter bound becomes $\frac{1}{2} - o(1)$. Therefore, by independence of the entries of $\mathbf{R}$, $Y_t$ stochastically dominates a binomial random variable ${\cal B}(t-1, \frac{1}{3})$. Furthermore, by the multiplicative Chernoff (upper) bound, for any $\delta>0$,
\begin{equation*}
\Pr\left(Y_t<(1-\delta) \frac{t-1}{3} \right) < \left( \frac{e^{-\delta}}{(1-\delta)^{1-\delta}}\right)^{\frac{t-1}{3}}.
\end{equation*}


Taking $\delta = \frac{1}{2}$ and noting that $t \geq \epsilon n +1$, we have
\begin{equation*}
\Pr\left(Y_t<\frac{t-1}{6} \right) < \left( \frac{e}{2}\right)^{-\frac{\epsilon n}{6}},
\end{equation*}
which is $o(1/n)$, for any constant $\epsilon >0$. By the union bound, 
\begin{equation*}
\Pr\left(\exists t: t \geq \epsilon n+1, Y_t<\frac{t-1}{6} \right) =o(1).
\end{equation*}
By inequality (\ref{ineq:N_tY_t}), we thus have that, with high probability over the choices of $\mathbf{R}$, $N_t \geq \frac{t-1}{12} \geq \frac{\epsilon n}{12}$, for all $t \geq \epsilon n+1$, as needed.  

Combining the above, by the Berry-Esseen Theorem, given $\mathbf{x}_{[t-1]}, \mathbf{R}_{[m], [t-1]}$, the distribution of $Z_t^B-Z'^B_t$ is approximately Normal ${\cal N}(0, 2p(1-p)N_t)$ with 
approximation error $o(1)$ as $c \to \infty$, with high probability over the choices of $\mathbf{R}$. In particular, given $\mathbf{x}_{[t-1]}, \mathbf{R}_{[m], [t-1]}$, $|Z_t^B-Z'^B_t|$ follows approximately (i.e. with the same approximation error $o(1)$) the \emph{folded normal distribution} with mean value (at least) $\sqrt{\frac{2}{\pi} \text{Var}(Z_t^B-Z'^B_t| \mathbf{x}_{[t-1]}, \mathbf{R}_{[m], [t-1]})}$. Notice now that, by inequality (\ref{ineq:N_tY_t}), we have
\begin{equation*}
\text{Var}(Z_t^B-Z'^B_t| \mathbf{x}_{[t-1]}, \mathbf{R}_{[m], [t-1]}) \geq p(1-p) Y_t.
\end{equation*}

Since $Y_t \geq \frac{t-1}{6} \geq \frac{\epsilon n}{6}$ with high probability, and also $p = \frac{c}{n}$, we get that $\text{Var}(Z_t^B-Z'^B_t| \mathbf{x}_{[t-1]}, \mathbf{R}_{[m], [t-1]}) \geq \frac{c (t-1)}{6n} -o(1)$, with high probability, where the $o(1)$ comes from the approximation error given by the Berry-Esseen Theorem. Consequently, by inequality (\ref{ineq:MD}), with high probability over the choices of $\mathbf{R}$ (which is $1- o(1)$),
\begin{equation*}
\mathbb{E}\left[|Z_t| \right] = \mathbb{E}\left[\left| \sum_{i \in [t-1]} \left[\mathbf{R}^T \mathbf{R} \right]_{i,t} x_i\right| \right] \geq \sqrt{\frac{c (t-1)}{3 \pi n}} -o(1).
\end{equation*} 

Summing over all $t \geq \epsilon n+1$, we get
\begin{equation*}
\sum_{t \geq \epsilon n+1} \mathbb{E}\left[|Z_t| \right] \geq \sqrt{\frac{c}{3 \pi n}} \sum_{t \geq \epsilon n} \sqrt{t} -o(n) = \sqrt{\frac{c}{3 \pi n}} \left(\sum_{t \geq 1} \sqrt{t} - \epsilon n \sqrt{\epsilon n} \right) -o(n).
\end{equation*} 
Using the fact that $\sum_{t \geq 1} \sqrt{t} = \frac{2}{3} n^{3/2} +o(n)$, we thus have that 
\begin{equation*}
\sum_{t \geq \epsilon n+1} \mathbb{E}\left[|Z_t| \right] \geq \sqrt{\frac{c}{3 \pi}} \left(\frac{2}{3} - \epsilon ^{3/2}\right) n -o(n).
\end{equation*} 

On the other hand, we have that the expected weight of a random cut is equal to $\frac{1}{4} n(n-1)mp^2 = \frac{c^2}{4}n + o(n)$ (see e.g. equation (\ref{eq:randomcut})). The proof is completed by taking $\epsilon \to 0$.
\end{proof}

\section{Proof of Lemma \ref{lem:disjointcycles}} \label{sec:lemma1proof}

\begin{proof}
We will use the first moment method and so we need to prove that the expectation of the number of pairs of distinct 0-strong closed vertex-label sequences in $G$ that have at least one label in common goes to 0. To this end, for $j \in [\min(k, k')-1]$, let $A_j(k, k')$ denote the number of such sequences $\sigma, \sigma'$, with $k=|\sigma|, k' = |\sigma'|$, that have $j$ labels in common. In particular, for integers $k, k'$, let $ \sigma:=v_1, \ell_1, v_2, \ell_2, \cdots, v_k, \ell_{k}, v_{k+1}=v_1$, and let $\sigma' :=v'_1, \ell'_1, v'_2, \ell'_2, \cdots, v'_{k'}, \ell'_{k'}, v'_{k'+1}=v_1$. Notice that, any such fixed pair $\sigma, \sigma'$ has the same probability to appear, namely $p^{2(k+k'-j)} (1-p)^{(n-2)(k+k'-j)}$; indeed, $p^{2k} (1-p)^{(n-2)k}$ is the probability that $\sigma$ appears (recall that $\sigma$ has $k$ labels and it is 0-strong, i.e. each label is only selected by two vertices) and $p^{2(k'-j)} (1-p)^{(n-2)(k'-j)}$ is the probability that $\sigma'$ appears given that $\sigma$ has appeared. Furthermore, the number of such pairs of sequences is dominated by the number of sequences that overlap in $j$ consecutive labels (e.g. the first $j$), which is at most $n^k m^k n^{k'-j-1} m^{k'-j}$ (notice that $j$ common labels implies that there are at least $j'+1$ common vertices). Overall, since $n=m$ and $p = \frac{c}{n}$, we have
\begin{eqnarray*}
\mathbb{E}[A_j(k, k')] & \leq & (1+o(1)) \frac{1}{n} (np)^{2(k+k'-j)} (1-p)^{(n-2)(k+k'-j)} \\
& = & (1+o(1)) \frac{1}{n} \left( c^2 (1-p)^{n-2} \right)^{k+k'-j}.
\end{eqnarray*}
Since $n \to \infty$ and $p = \frac{c}{n}$, by elementary calculus we have that $c^2 (1-p)^{n-2}$ bounded by a constant (which depends only on $c$) strictly less than 1. Therefore, the above expectation is at most $e^{-\ln{n} - \Theta(1) (k+k'-j)}$. Therefore, summing over all choices of $k, k' \in [n]$ and $j \in [\min(k, k')-1]$, we get that the expected number of pairs of distinct 0-strong closed vertex-label sequences that have at least one label in common is at most

\begin{equation*}
\sum_{k, k' \in [n]} \sum_{j \in [\min(k, k')-1]} e^{-\ln{n} - \Theta(1) (k+k'-j)} = o(1),
\end{equation*}
and the proof is completed by Markov's inequality.
\end{proof}

\section{Proof of Lemma \ref{lem:only0oddcycles}} \label{sec:only0oddcycles}

\begin{proof} For the sake of contradiction, assume ${\cal C}_{odd}(G^{(b)}) = \emptyset$, but $G^{(b)} = \cup^+_{\ell \in [m]} M^{(\ell)}$ has an odd cycle $C_k$ that is not 0-strong and has minimum length. Notice that $C_k$ corresponds to a closed vertex-label sequence, say $\sigma:= v_1, \ell_1, v_2, \ell_2, \cdots, v_k, \ell_{k}, v_{k+1}=v_1$, where $\{v_i, v_{i+1}\} \in M^{(\ell_i)}$, for all $i \in [k]$. Furthermore, by assumption, conditions (b) and (c) of Definition \ref{def:codd} are satisfied by $\sigma$ (indeed $\{v_i, v_{i+1}\} \in M^{(\ell_i)}$, for all $i \in [k]$, and $\sigma$ is $\lambda$-strong, for some $\lambda>0$). Therefore, the only reason for which $\sigma$ does not belong to ${\cal C}_{odd}(G^{(b)})$ is that condition (a) of Definition \ref{def:codd} is not satisfied, i.e. there are distinct indices $i > i' \in [k]$ such that $\ell_i = \ell_{i'}$. Clearly, such indices are not consecutive (i.e. $i' \neq i+1$), because $\ell_i$ is strong and step 6 of our algorithm implies that $M^{(\ell_i)}$ is a matching of $K^{(\ell_i)}$. But then either the vertex-label sequence $v_1, \ldots, v_i, \ell_i, v_{i'+1}, \ell_{i'+1}, v_{i'+2}, \ldots, v_{k+1} = v_1$ or the vertex-label sequence $v_{i+1}, \ell_{i+1}, v_{i+2}, \ldots, v_{i'}, \ell_{i}, v_{i+1}$ corresponds to a shorter odd cycle, which is a contradiction on the minimality of $C_k$.
\end{proof}

\end{document}